\RequirePackage{silence}

\documentclass[aps,prx,twocolumn,10pt,nofootinbib,superscriptaddress]{revtex4-2}

\usepackage[utf8]{inputenc}
\usepackage[T1]{fontenc} 
\usepackage[english]{babel} 

\usepackage[table,usenames,dvipsnames]{xcolor}

\usepackage{amsmath}
\usepackage{amssymb,mathtools,amsthm}
\usepackage{dsfont} 
\usepackage{booktabs}
\usepackage{qcircuit}

\usepackage[printonlyused,withpage]{acronym}
\makeatletter
\AtBeginDocument{%
  \renewcommand*{\AC@hyperlink}[2]{%
    \begingroup
      \hypersetup{hidelinks}%
      \hyperlink{#1}{#2}%
    \endgroup
  }%
}
\makeatother

\makeatletter
\def\l@subsubsection#1#2{}
\makeatother

\usepackage{paralist}
\usepackage{enumitem}

\newcommand{\func}[1]{{\ensuremath{\mathsf{#1}}}}

\newcommand{\alg}{\func{mGST}}
\newcommand{\pyGSTi}{\func{pyGSTi}}

\usepackage{xifthen}

\usepackage{import}

\usepackage{scrhack}
\usepackage[linesnumbered,lined,ruled,commentsnumbered]{algorithm2e}
\usepackage[noend]{algpseudocode}

\usepackage[colorlinks=true,
      hyperindex, 
        linkcolor = NavyBlue,
        anchorcolor = hyperrefblue,
        citecolor = Green,
        filecolor = hyperrefblue,
        urlcolor = NavyBlue,
      breaklinks=true]{hyperref}

\usepackage{bookmark} 

\providecommand{\nat}{Nature}

\providecommand{\pra}{Phys.\ Rev.\ A}
\providecommand{\prb}{Phys.\ Rev.\ B}

\providecommand{\prl}{Phys.\ Rev.\ Lett.}

\providecommand{\njp}{New J.\ Phys.}

\newtheorem{thm}{Theorem}[section]
\newtheorem{lem}[thm]{Lemma}

\newcommand{\ngates}{n}
\newcommand{\nsamples}{m}
\newcommand{\seqlength}{\ell}
\newcommand{\Lo}{\mathcal{L}}
\newcommand{\vvec}{\operatorname{vec}}
\DeclareMathOperator{\St}{St}

\newtheorem*{problem}{Problem}
\theoremstyle{definition}

\newcommand{\e}{\ensuremath\mathrm{e}}
\renewcommand{\i}{\ensuremath\mathrm{i}}
\newcommand{\rmd}{\ensuremath\mathrm{d}}

\DeclareMathOperator{\LandauO}{\mathrm{O}}

\DeclareMathOperator{\Tr}{Tr}
\renewcommand{\Re}{\operatorname{Re}}

\DeclareMathOperator{\Id}{Id}

\newcommand{\fro}{\mathrm{F}}

\renewcommand{\L}{\operatorname{L}}

\DeclareMathOperator{\GL}{GL} 

\newcommand{\CC}{\mathbb{C}}

\newcommand{\EE}{\mathbb{E}}

\newcommand{\mc}[1]{\mathcal{#1}}

\newcommand{\GG}{\mc{G}}

\newcommand{\K}{\mc{K}}

\newcommand{\argdot}{{\,\cdot\,}}
\renewcommand{\vec}[1]{\boldsymbol{#1}}

\newcommand{\myleft}{\mathopen{}\mathclose\bgroup\left}
\newcommand{\myright}{\aftergroup\egroup\right}

\newcommand{\abs}[1]{\left\vert #1 \right\vert} 

\newcommand{\norm}[1]{\left\Vert #1 \right\Vert} 
\newcommand{\pnorm}[2][p]{\norm{#2}_{#1}} 

\newcommand{\iiiNorm}[1]{{\left\vert\kern-0.25ex\left\vert\kern-0.25ex\left\vert #1 
    \right\vert\kern-0.25ex\right\vert\kern-0.25ex\right\vert}}
\newcommand{\fnorm}[1]{\norm{#1}_\fro} 

\newcommand{\braket}[2]{\left\langle #1 \middle| #2 \right\rangle}

\newcommand{\ket}[1]{\left.\left|{#1}\right.\right\rangle}

\newcommand{\bra}[1]{\left.\left\langle{#1}\right.\right|}

\newcommand{\ketr}[1]{\left.\left|{#1}\right.\right)}
\newcommand{\brar}[1]{\left.\left({#1}\right.\right|}

\newcommand{\braketr}[2]{\left.\left(#1 \middle| #2 \right.\right)}

\newcommand{\ketbra}[2]{\ket{#1} \!\! \bra{#2}}

\newcommand{\sandwich}[3]
  {\left\langle  #1 \right| #2 \left| #3 \right\rangle}

\DeclarePairedDelimiter{\obra}{(}{\vert}
\DeclarePairedDelimiter{\oket}{\vert}{)}

\DeclarePairedDelimiterX\obraket[2]{(}{)}%
  {#1\kern0.15ex\delimsize\vert\kern0.15ex\mathopen{}#2}

\DeclarePairedDelimiterX\oketbra[2]{\vert}{\vert}%
  {#1\kern0.15ex\delimsize)\delimsize(\kern0.15ex\mathopen{}#2}

\DeclarePairedDelimiterX\osandwich[3]{(}{)}%
  {#1\,\delimsize\vert\kern0.15ex\mathopen{}#2\kern0.15ex\delimsize\vert\kern0.15ex\mathopen{}#3}

\usepackage{tikz}
\usetikzlibrary{decorations.pathreplacing,calc}

\newcommand{\tikzmark}[1]{\tikz[overlay,remember picture] \node (#1) {};}

\newcommand*{\AddNote}[4]{%
    \begin{tikzpicture}[overlay, remember picture]
        \draw [decoration={brace,amplitude=0.5em},decorate]
            ($(#3)!(#1.north)!($(#3)-(0,1)$)$) --  
            ($(#3)!(#2.south)!($(#3)-(0,1)$)$)
                node [align=center, text width=1.6cm, pos=0.5, anchor=west] {#4};
    \end{tikzpicture}
}%

\makeatletter 
 \hypersetup{pdftitle = {Compressive gate set tomography},
       pdfauthor = {Raphael Brieger, Ingo Roth, Martin Kliesch},
       pdfsubject = {Quantum physics}, 
       pdfkeywords = {quantum, state, preparation, process, measurement, tomography, 
              characterization, computer, benchmark,
              SPAM, robust, stable, 
              matrix product state, MPS, GST, gate, set, compressive, compressed, sensing,
              structured, TT, tensor, tensor train, completion, reconstruction, 
              gradient descent, second order, 
              alternating, empirical, risk, minimization, least squares, total variation, distance, mean variation, error
              saddle-free, Newton, method, saddle, point, Hessian, 
              manifold, Riemannian, constrained, optimization, 
              gauge, freedom, CPT, completely, positive, trace, preserving, constraint             
              }
      }
\makeatother

\newcommand{\hhu}{%
  Institute for Theoretical Physics,
  Heinrich Heine University D{\"u}sseldorf, 
  Germany%
}
\newcommand{\tuhh}{
Institute for Quantum-Inspired and Quantum Optimization, 
Hamburg University of Technology, 
Germany
}
\newcommand{\fu}{%
  Dahlem Center for Complex Quantum Systems,
  Freie Universit\"{a}t Berlin,
  Germany%
}
\newcommand{\AbuDhabi}{%
Quantum Research Centre, Technology Innovation Institute, Abu Dhabi, UAE
}

\definecolor{martin}{rgb}{0,.4,1}

\definecolor{ingo}{rgb}{.1,.5,.1}

\definecolor{rb}{HTML}{018B86}

\begin{document}
\title{Compressive gate set tomography}

\author{Raphael Brieger}
\email{brieger@hhu.de}
\affiliation{\hhu}

\author{Ingo Roth}
\affiliation{\AbuDhabi}
\affiliation{\fu}

\author{Martin Kliesch}
\email{martin.kliesch@tuhh.de}
\affiliation{\hhu}
\affiliation{\tuhh}

\begin{abstract}
Flexible characterization techniques that provide a detailed picture of the experimental imperfections under realistic assumptions are crucial to gain actionable advice in the development of quantum computers. Gate set tomography self-consistently extracts a complete tomographic description of the implementation of an entire set of quantum gates, as well as the initial state and measurement, from experimental data. 
It has become a standard tool for this task but comes with high requirements on the number of sequences and their design, making it experimentally challenging already for only two qubits.

In this work, we show that low-rank approximations of gate sets can be obtained from significantly fewer gate sequences and that it is sufficient to draw them at random. This coherent noise characterization however still contains the crucial information for improving the implementation. To this end, we formulate the data processing problem of gate set tomography as a rank-constrained tensor completion problem. We provide an algorithm to solve this problem while respecting the usual positivity and normalization constraints of quantum mechanics. 
For this purpose, we combine methods from Riemannian optimization and machine learning and develop a saddle-free second-order geometrical optimization method on the complex Stiefel manifold. 
Besides the reduction in sequences, we demonstrate numerically that the algorithm does not rely on structured gate sets or an elaborate circuit design to robustly perform gate set tomography. 
Therefore, it is more flexible than traditional approaches. 
We also demonstrate how coherent errors in shadow estimation protocols can be mitigated using estimates from gate set tomography.
\end{abstract}

\maketitle

\tableofcontents

\section{Introduction}
The precise characterization of digital quantum devices is crucial for several reasons:
(i) to obtain `actionable advice' on how imperfections 
on their implementation can be reduced, e.g.\ by experimental control, 
(ii) to tailor applications to unavoidable device errors so that their effect can be mitigated, and 
(iii) to benchmark 
the devices
for the comparison of different physical platforms and implementations. 
There is already a wide variety of protocols to characterize components of a digital quantum computing device  
with a trade-off between the information gained about the system and the associated resource requirements and assumptions of the scheme \cite{Eisert2020QuantumCertificationAnd, Kliesch2020TheoryOfQuantum}. 

One particular important requirement for practical characterization protocols for quantum gates is their robustness against errors in the \ac{SPAM}. 
There are two general approaches that \ac{SPAM}-robustly characterize the implementation of entire \emph{gate sets} of a quantum computer. 
On the low complexity side there is \ac{RB} \cite{EmeAliZyc05,KniLeiRei08,Magesan2012} and variants thereof \cite{Helsen20AGeneralFramework}, that typically aim at determining a single measure of quality for an experiment, though with the exception of \ac{RB} tomography protocols \cite{KimSilRya14,KimLiu16,RotKueKim18,Flammia2019EfficientEstimation}. 
Yet for the targeted improvement of individual quantum operations, 
protocols which provide more detailed information beyond mere benchmarking are crucial. 

This is the motivation of self-consistent \ac{GST} \cite{MerGamSmo13,Stark2012Self-consistentTomography,Stark2012SimultaneousEstimation,BluGamNie13,BluGamNie17,Gre15,Nielsen20ProbingQuantumProcessor,Nielsen2020GateSetTomography}. 
\ac{GST} estimates virtually all parameters describing a noisy implementation of a quantum computing device simultaneously from the measurements of many gate sequences. 
This comprises tomographic estimates for all channels implementing the gate set elements, the initial state(s) and the measurement(s). 
The full tomographic information can then be used to compute arbitrary error measures for verification and to provide advice on error mitigation and device calibration ~\cite{BluGamNie17,dehollain2016optimization,2021PRXQ....2d0338R,2019SciA....5.5686S,2020NatCo..11..587Z,2020arXiv200701210C}.
Concomitant with the massive amount of inferred information and minimal assumptions, 
these protocols come with enormous resource requirements 
in terms of the necessary number of measurement rounds 
and 
the time and storage consumption of the classical post-processing. 
Standard \ac{GST}, as described, e.g.\ by Nielsen et al.\ \cite{Nielsen2020GateSetTomography}, uses many carefully designed gate sequences in the experiment and a sophisticated and challenging 
data processing pipeline in post. 
To arrive at physically interpretable estimates, i.e.\ \ac{CPT} maps, additional post-processing is required. 
The massive amount of 
specific data consumed by standard \ac{GST} 
limits its practical applicability already for two-qubit gate sets. 
The focus of Nielsen et al.\ \cite{Nielsen2020GateSetTomography} and their implementation \pyGSTi\ lies on so-called \emph{long sequence \ac{GST}}, a method to improve an initial \ac{GST} estimate by using gate sequences in which a building block is repeated many times. 
The resulting error amplification of the building block is then used to significantly improve the accuracy of the \ac{GST} estimate at the cost of larger measurement effort. 
In our work we focus on \emph{short sequences} and the problem of finding an initial estimate without assuming any prior knowledge on the gate set and minimal experimental requirements.

The most important diagnostic information for a quantum computing device is often already 
contained in a low-rank approximation of the processes, states and measurements.  
Coherent errors are typically the ones that can be corrected by experimental control and are of interest for refining calibration models. 
The strength of incoherent noise on the other hand is arguably well-captured by average error measures as provided by \ac{RB} outputs. 
Moreover, current fault tolerance thresholds often rely on 
worst-case error measures for which no good direct estimation technique exists \cite{knill05,Aliferis2006a,Aliferis2006a,Cross07ComparativeCode,2009PhRvA..79a2332A} and coherent errors 
in particular hinder their indirect inference from average error measures \cite{KueLonFla15,Sanders2015,Wal15}. 
For standard state and process tomography, it was realized that 
low-rank assumptions can crucially reduce the sample complexity, the required number of measurements  and the post-processing complexity \cite{ShaKosMoh11a,FlaGroLiu12,KliKueEis16,KliKueEis19,BalKalDeu14,KueRauTer15,RodVeiBar14,RioGroFla16,SteRioCut16,RotKueKim18,KyrKalPar18} as well as improve the stability against imperfections in the measurements \cite{RothEtAl:2020:SDD} using compressed sensing techniques \cite{gross_quantum_2010,Gro11,FouRau13}. 

In this work, we take a fresh look at the data processing problem of \ac{GST} from a compressed sensing perspective and regard it as a highly-structured tensor completion problem.
We develop a reconstruction method, called \alg, that 
exploits the geometric structure of \ac{CPT} maps with low Kraus ranks. 
In numerical simulations we demonstrate that our structure-exploiting \alg\ approach (i)~allows for 
maximal flexibility in the design of gate sequences, so that standard \ac{GST} gate sequences and random sequences work equally well, and 
(ii)~obtains low-rank approximations of the implemented gate set from a significantly reduced number of sequences and samples. 
This allows us to successfully perform \ac{GST} with gate sets and sequences that are not amenable to the standard \ac{GST} implementation \pyGSTi\ \cite{Nielsen20pyGSTi,Nielsen20ProbingQuantumProcessor}.
As one example, while the sequence design of \pyGSTi\ uses at least $907$ specific sequences to reconstruct a two-qubit gate set, we numerically demonstrate low-rank reconstruction from $200$ random sequences of maximal length $\ell=7$ with runtimes of less than an hour on a standard desktop computer. 
Thus, compressive \ac{GST} significantly lowers the experimental resource requirements for maybe the most prominent use-cases of \ac{GST} making it a tool that can be more easily and routinely applied. 
At the same time, for the default gate sets and sequences from the standard \ac{GST} implementation, the novel algorithm matches state-of-the-art results. 
The runtime and storage requirements of \alg\ still scale exponentially in the number of qubits as does the amount parameters of the gate set it identifies.  
This limits the feasibility of the classical post-processing of compressive \ac{GST} to gate sets acting on only a few qubits without further assumptions. 
Nonetheless, we demonstrate that coherent errors and depolarizing noise 
of a 3-qubit gate set can be completely characterized, from as little as $128$ sequences of length $\ell=7$ on desktop hardware in a few hours.  

To give a novel example of how information about coherent errors can be used, we simulate a 10 qubit system and perform GST on neighboring 2-qubit pairs. The resulting gate set estimates then allow us to calibrate the post-processing step of the shadow estimation protocol \cite{Huang2020Predicting}, which is widely used for the sample efficient estimation of observables. More concretely, we find in simulations with moderate coherent errors that shadow estimates of the ground state energy of a 10 qubit Heisenberg Hamiltonian are heavily biased when knowledge of the noisy gate implementation is limited. Information from GST on 2 qubit pairs allows us to reduce this bias by about an order of magnitude. 

Our \alg\ reconstruction method relies on manifold optimization over complex Stiefel manifolds
\cite{Abrudan2009ConjugateGradient,edelman1998geometry,manton2002optimization,sun2019escaping,bortoloti2020damped,wisdom2016full,boumal2019global} in order to include the low-rank \ac{CPT} constraints. 
Such constraints emerge in several optimization problems \cite{helmke2012optimization,1166684,elden1999procrustes,BRIDGES2001219} with applications in machine learning \cite{wisdom2016full,Lotte_2018}, quantum chemistry \cite{CHIUMIENTO20121866}, signal processing in wireless communication \cite{7397861,8288677} and more recently in the quantum information literature, see e.g.\ \cite{luchnikov2020riemannian,Hangleiter21PreciseHamiltonianIdentification,HangleiterEtAl:2020:Easing}. 
In order to deal with the non-convex optimization landscape we adopt a second order saddle-free Newton method \cite{Dauphin2014IdentifyingAndAttacking} to this setting. 
This involves the derivation of an analytic expression for geodesics, as well as an expression for the Riemannian Hessian in the respective product manifolds.
Another important motivation for phrasing \ac{GST} as a randomly subsampled tensor completion problem is to bring it closer to potential analytical recovery guarantees common for related tensor completion problems \cite{SuessPhDThesis,ImaMaeHay17,RauSchSto16,GhaPlaYil17,Ashraphijuo17CharacterizationOfDeterministic,RauSto15,HuaMuGol14,liu2020tensor}, opening up a new research direction. 

Finally, being able to perform \ac{GST} from random sequences enables one to use the same type of data for different increasingly refined characterization tasks from filtered \ac{RB} \cite{Helsen20AGeneralFramework}, \ac{XEB} \cite{BoiIsaSme16} and \ac{RB} tomography \cite{KimSilRya14,KimLiu16,RotKueKim18} to \ac{GST}. 
Unifying these approaches, random gate sequences can be regarded as the `classical shadow' of a gate set from which many properties can be estimated efficiently \cite{HelsenEtAl:2021:Estimating}.  
Compressive gate set tomography provides more detailed diagnostic information and only requires to further increase the amount of data without changing the experimental instructions.

With randomized linear \ac{GST} \cite{Gu2021RandomizedLinearGate} and fast Bayesian tomography \cite{2021arXiv210714473E} 
related alternatives to tackle the \ac{GST} data processing problem 
have been proposed. 
Here, the gates are assumed to be well-approximated by an a priori known unitary followed by a noise channel that is either linearized around the identity \cite{Gu2021RandomizedLinearGate} or around a prior noise estimate \cite{2021arXiv210714473E}. 
This allows for a treatment of the outcome probabilities as approximately linear functions. 
The resulting scheme already works for random sequence data but comes at the expense of much stronger assumptions compared to the compressed sensing approach taken with \alg. 

The rest of the paper is structured into three parts. 
In the following section, we 
formalize the data processing problem of \ac{GST} as a constrained reconstruction problem. 
In Section~\ref{sec:data processing via Riemannian optimization}, we formulate the data processing problem as a geometric optimization task and derive 
the \alg\ algorithm. 
In Section~\ref{sec:Numerical analysis} we demonstrate the performance of the novel algorithm in numerical simulations and compare our results with the standard \ac{GST} processing pipeline of \pyGSTi.

\section{The data processing problem of gate set tomography}

In \ac{GST} a quantum computing device is modeled as follows. 
The device is initialized with a state $\rho \in \mathcal{S} \coloneqq \{\sigma \in \mathcal{H}: \sigma \succeq 0, \Tr[\sigma] = 1\}$ on a finite dimensional Hilbert space $\mathcal{H} = \mathbb{C}^d$. 
Subsequently, a sequence of noisy operations from a fixed gate set $(\GG_i)_{i \in [n]}$ can be applied, where we use the notation $[n] \coloneqq \{1,2,3,\dots,n\}$. 
The noisy operations $\GG_i: \L(\mathcal{H}) \rightarrow \L(\mathcal{H})$ are \ac{CPT} maps on $\L\left(\mathcal{H}\right)$, the set of linear operators on $\mathcal{H}$. 
We define $[n]_\seqlength^* \coloneqq \bigcup_{k=0}^{\seqlength}[n]^k$ such that $\vec i \in [n]_\seqlength^*$ defines a gate sequences with length of at most $\seqlength$ 
 and associated \ac{CPT} map $\GG_{\vec i} \coloneqq \GG_{i_{\seqlength}}\circ\dots\circ\GG_{i_1}$, the concatenation of the gates in the sequence $\vec i$.
In the end, a measurement is performed described by a \ac{POVM} with elements $(E_j)_{j \in [n_E]}$, satisfying $\sum_j E_j = \mathds{1}$ and $0 \preceq E_j \preceq \mathds{1}$ for all $j \in [n_E]$.
The full description of the noisy quantum computing device is, thus, given by the triple 
\begin{equation}\label{eq:def:X}
  \mc X = ((E_j)_{j \in [n_E]}, (\GG_i)_{i \in [n]}, \rho)
\end{equation}
of a quantum state, a physical gate set and a \ac{POVM}.
Let $I \subseteq [n]_\seqlength^*$ be the set of accessible gate sequences with $n_{\mathrm{seq}} \coloneqq |I|$ denoting the number of sequences.
The probability of measuring outcome $j$ upon applying a gate sequence $\vec i \in I$ is
\begin{equation} \label{eq:pijX}
   p_{j|\vec i}(\mc X) = \Tr [E_j \, \GG_{\vec i}(\rho)] \, .
\end{equation}
By $(\vec p_{\vec i}(\mc X))_j \coloneqq p_{j|\vec i}(\mc X)$ we denote the corresponding vector and, moreover, often omit the argument $\mc X$. 
While being a fairly general description, this gate set model relies on a couple of assumptions:
\begin{compactenum}[(i)] 
\item the physical system needs to be well-characterized by a Hilbert space of fixed dimension, 
\item the system parameters need to be time independent over different experiments, and
\item a gate's action is independent of the gates applied before and after (Markovianity). 
\end{compactenum}

There exist multiple descriptions of quantum computing devices within the gate set model that yield the same measurement probabilities on all sequences. 
Below, we provide a more detailed description of this freedom in terms of \emph{gauge transformations}. 
These are linear transformations that, when simultaneously applied to all gates, input state and \emph{POVM} elements, leave the measurement statistics \eqref{eq:pijX} invariant. 

The task of \ac{GST} is to infer the device's full description \eqref{eq:def:X} from measured data.  
To this end, one estimates the output probabilities for a set of different sequences $I \subset [n]_l^\ast$ by repeatedly performing the measurements of the corresponding sequences. 
Thus, we can state the \emph{data-processing problem of \ac{GST}} as follows. 

\begin{problem}[\ac{GST} data-processing]
Let $\mc X$ be a gate set and $I \subset [n]_l^\ast$ a set of sequences. 
Given empirical estimates $\{y_{j|\vec i}\}_{\vec i \in I, j \in [n_E]}$ of 
$\{p_{j|\vec i}(\mc X)\}_{\vec i \in I, j \in [n_E]}$, find the device description $\mc X = ((E_j)_{j \in [n_E]}, (\GG_i)_{i \in [n]}, \rho)$ up to the gauge freedom.
\end{problem}

Note that \ac{GST} aims at solving an \emph{identification problem}. 
That is, for sufficiently much data, find the unique device description of the device compatible with the data. 
In particular, the input data $\{\hat p_{j|\vec i}\}_{\vec i \in I, j}$ is required to uniquely single out the device description. 
This is related but distinct from the corresponding \emph{learning task} to find a description that generalizes on unseen data.

\subsection{Compressive gate set description} 
\label{sec:phys. constr.}

At the heart of our approach is to capture this data-processing problem as a highly structured \emph{tensor completion problem}. 
The structure allows us to reduce the required size and structural assumption of the set $I$, in order to determine $\mc X$.
It is instructive to visualize the problem with tensor network diagrams. 
The gate set can be viewed as a tensor of five indices and the action of gate $i$ on the initial state $\rho$ can be visualized as \vspace{.5\baselineskip}
    \begin{equation}
    \GG_i(\rho) = \quad 
      \begin{minipage}{.3\linewidth}
      \vspace{-\baselineskip}
      \includegraphics{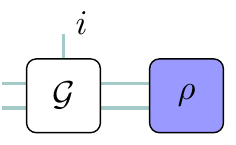}
      \end{minipage}
      \begin{minipage}{2pt}
      $\, $\\[.3em]
      ,\end{minipage}
    \end{equation}
where each leg represents an open index and the joining of legs represents summation of the corresponding indices; 
see e.g.\ \cite[Chapter~5.1]{montangero2018introduction} for more information on the tensor network notation. 
Neglecting the finite statistics in estimating the probabilities, 
the \ac{GST} data-processing problem can, thus, be rephrased as the 
problem of completing the translation-invariant \emph{\ac{MPS}} \cite{FanNacWer92,Rommer97ClassOfAnsatz,Verstraete04DMRG} or \emph{tensor-train} \cite{Oseledets09Breaking}
\begin{align}
p_{j|\vec i} &= \Tr [E_{j} \, \GG_{i_l}\circ \dots \circ \GG_{i_2}\circ\GG_{i_1}(\rho)] \label{eq:probabilities}
=
\\[.2em] \nonumber
& \phantom{=}
  \includegraphics{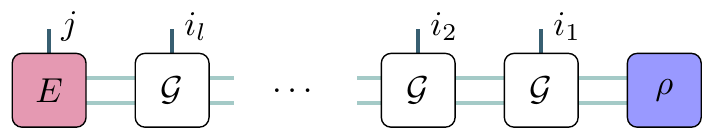}
\end{align}
from access to a couple of its entries. 
By the following assumptions one can introduce more structure. 
First, we assume the elements of the device description $\mc X$ to satisfy the physicality assumptions regarding normalization and positivity. 
Second, the assumption that they have low-rank approximations yields additional compressibility of $\mc X$.  

In more detail, physically implementable gate sets are \acl{CPT} if and only if they admit a Kraus decomposition \cite{KrausEtAl:1983:States}, i.e.\ the $i$-th gate implementation can be written as 
\begin{equation}
  \GG_i(\rho) = \sum_{l=1}^{r_K} \K_{il}\rho \K_{il}^{\dagger} 
\end{equation}
for each $i\in [n]$, 
where $r_K$ is the (maximum) Kraus rank of the \ac{CPT} maps $\{\mc G_i\}$. 
We use the notation that $\K$ denotes the tensor containing all Kraus operators of all gates and $\K_i$ contains the Kraus operators for gate $i$.
In terms of tensor network diagrams the decomposition is represented as 
\begin{equation}
  \begin{minipage}{.44\linewidth}
  \includegraphics{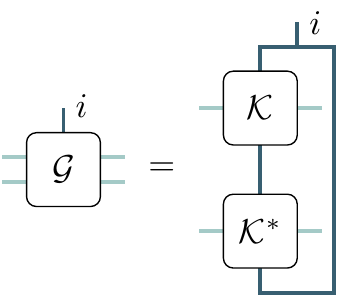}
  \end{minipage}
  \begin{minipage}{2pt}
  $\ $\\[-.2em]
  .\end{minipage}
\end{equation}
Moreover, the trace preservation constraints 
\begin{equation}\label{eq: Kraus constraint}
   \sum_{l=1}^{r_K} \K_{il}^{\dagger} \K_{il} = \mathds{1} \qquad \forall i
\end{equation} 
require $\K_i$ viewed as a matrix in $\mathbb{C}^{r_K d\times d}$ to be an isometry, i.e.\ 
\begin{equation}
  \begin{minipage}{.34\linewidth}
  \includegraphics{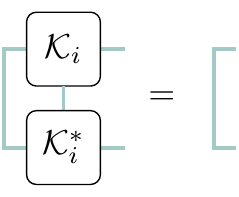}
  \end{minipage}
  \begin{minipage}{2pt}
  $\ $\\[-.2em]
  .\end{minipage}
\end{equation}
Constraints on a low Kraus rank $r_K$ can be naturally enforced in this parametrization by reducing the row dimension of $\K_i$. 
The initial state and \ac{POVM} elements are constrained to be positive matrices, which we hence parameterize as 
\begin{align}\label{eq:def:AB}
  E_j = A_j^{\dagger}A_j, \qquad \rho = BB^{\dagger}
\end{align}
with $A_{j} \in \mathbb{C}^{r_E \times d}$ and $B \in \mathbb{C}^{d \times r_{\rho}}$, where $r_E$ and $r_{\rho}$ are the matrix ranks. 
For the matrices $A_j$ to form a valid \ac{POVM}, they have to satisfy a similar condition to the Kraus operators, 
\begin{equation} \label{eq:POVM constr.}
  \sum_{j=1}^{r_E} A_j^{\dagger}A_j = \mathds{1} \, ,
\end{equation}
while the initial state is of unit trace if
\begin{equation} \label{eq:state constr.}
  \fnorm{B} = 1\,.
\end{equation}

With the physicality constraints incorporated, measurement outcome probabilities are given in terms of tensor network diagrams as
\begin{align}
p_{j|\vec i} &= \Tr [E_{j} \, \GG_{i_l}\circ \dots \circ\GG_{i_2}\circ\GG_{i_1}(\rho)] 
=
\\[.2em] \nonumber
& \phantom{=}
  \begin{minipage}{.87\linewidth}
  \includegraphics{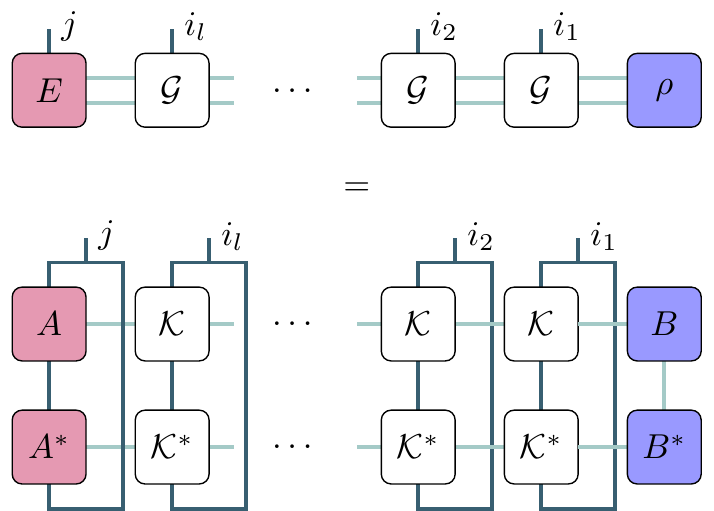}
  \end{minipage}
  \begin{minipage}{2pt}
  $\ $\\[-1em]
  .\end{minipage}
\end{align}

Thus, we arrive at a compressive device description 
$\mathcal{X}_c = (A, \mc K, B)$ 
that considerably reduces the amount of parameters compared to the triple 
$\mc X = ((E_j)_{j \in [n_E]}, (\GG_i)_{i \in [n]}, \rho)$
when choosing small dimensions $r_\rho$, $r_K$, and $r_E$. 
Correspondingly, we can adapt the \ac{GST} data processing task 
to demand only a compressive device description. 
\begin{problem}[Compressive \ac{GST} data processing]
Let $\mc X$ be a gate set and $I \subset [n]^\ast_l$ a set of sequences. 
Given empirical estimates 
$\{y_{j|\vec i}\}_{\vec i \in I,j \in [n_E]}$ of 
$\{p_{j|\vec i}(\mc X)\}_{\vec i \in I,j \in n_E}$ 
and ranks $r_\rho$, $r_K$, and $r_E$, find the compressive device description 
$\mc X_c = (A, \mc K, B)$ of dimension  $r_\rho$, $r_K$ and $r_E$ respectively, so that the normalization constraints \eqref{eq: Kraus constraint}, \eqref{eq:POVM constr.}, and \eqref{eq:state constr.} are satisfied. 
\end{problem}

As before, the set of sequences needs to be large enough so that this identification problem is well-defined. 
Again a desired compressive device description $\mc X_c$ can only be determined up to gauge freedom. 
Note that for the identification problem to be well-defined, it is not required that the true gate set $\mc X$ that generated the data is of low-ranks itself.  
As one usually aims to implement unit rank states, unitary gates, and basis measurements, i.e., for $r_\rho=r_K = r_E = 1$, it can be expected however that a compressive device description  $\mc X_c$ is often also a good approximation to the true gate set. Moreover, coherent errors are arguably the most relevant, since they give actionable advice on error mitigation and are complementary to the incoherent error measures provided by randomized benchmarking experiments.
By choosing the ranks in $\mc X_c$, a problem specific decision can be made that balances the information gained with the computational and sample complexity of model reconstruction. \\
Since the $p_{j|\vec i}(\mc X)$ are high degree polynomials in the gate set parameters, the compressive \ac{GST} data processing problem is different from compressed sensing for standard state and process tomography, where the map from the model parameters to the outcome probabilities is linear. \\
Next, we discuss another unique problem of \ac{GST}, the gauge freedom, more explicitly and introduce relevant error measures for gates sets $\mc X$.

\subsection{Gauge freedom and gate set metrics}
\label{sec:gauge_freedom}

So far we have not made explicit what `finding a device description' actually means. What is well studied in the \ac{GST} and \ac{RB} literature
\cite{Gre15,BluGamNie13,Lin19OnTheFreedom,Rudnicki2018gaugeinvariant,proctor2017WhatRandomizedBenchmarking,Helsen20AGeneralFramework}, is
that without additional prior assumptions, there is a freedom in representing a device in the gate set model. 
In particular, this freedom needs to be considered when defining a metric for gate sets \cite{Lin19OnTheFreedom} w.r.t.\ which we want to recover the device description. 

\emph{Gauge freedom} refers to the following observation. 
The observable measurement probabilities $p_{j|\vec i}$ of the form \eqref{eq:probabilities} are invariant under the transformation 
\begin{align}
\rho &\mapsto \mc T^{-1}(\rho)\\
\GG_i &\mapsto \mc T^{-1}\circ\GG_i\circ \mc T \qquad \forall i\\
E_j &\mapsto \mc T^\dagger(E_j)  \qquad \forall j
\end{align}
for any invertible super operator $\mc T: \L(\mathcal{H}) \rightarrow \L(\mathcal{H})$ 
, where $\mc T^\dagger$ denotes the adjoint of $\mc T$ w.r.t.\ the Hilbert-Schmidt inner product. 
This invariance is also the well-known \emph{gauge freedom of \ac{MPS}} \cite{Sch11}. 

If the gate set is universal and the initial state is pure then the gauge transformation $\mc T$ has to be either a unitary or anti-unitary channel. 
This statement can be seen as follows. 

In our case the $\GG_i$ are constrained to be \ac{CPT}. 
Hence, $\GG \mapsto \mc T^{-1}(\GG) \mc T$ has to map \ac{CPT} maps to \ac{CPT} maps. 
Similarly, $\mc T^{-1} \rho$ has to be a density operator and $\{\mc T^\dagger(E_j) \}_j$ a valid \ac{POVM}. 

A more explicit condition on $\mc T$ can be obtained by considering gauge action on entire sequences. 
For all sequences $\vec i$, we have
\begin{align*}
  p_{j|\vec i} = \Tr[E_j \mc T \circ \mc T^{-1}\circ \GG_{\vec i}(\rho)]
\end{align*}
where $\GG_{\vec i}(\rho)$ is a positive operator if the gates $\GG_i$ are \ac{CPT}. 
Now $\mc T^{-1}\GG_{\vec i} (\rho)$ has to be positive as well for all sequences $\vec i$. 
Thus if the gate set is universal, the map $\mc T^{-1}$ has to be positive and trace preserving for all states. 
An analogous statement can be made for $\mc T^\dagger$, by considering that
$\mc T^\dagger\GG_{\vec i}^\dagger(E_j)$ has to be positive-definite for all \ac{POVM} elements $E_j$. 
This implies that $\mc T$ has to be a positive map as well. 

It has been shown that any positive invertible map $\mc T$ with a positive inverse can be written either as 
$\mc T(\rho) = P U \rho\, U^{\dagger} P^\dagger$ or $\mc T(\rho) = P U \rho^T U^{\dagger} P^\dagger$ for 
$U \in \GL(d,\CC)$ and $P\succeq 0$  
\cite[Theorem~2]{schneider1965positive}. 
The condition that $\mc T$ needs to be trace preserving then yields $P =\mathds{1}$, as can be seen from the Kraus decomposition. 
Hence, $\mc T$ is indeed  either a unitary or anti-unitary channel. 

We note that the map $\rho \mapsto U \rho^T U^{\dagger}$ is positive but not completely positive. 
However, it has the property that $\mc T^{-1} \GG_i \mc T$ is \ac{CPT} whenever $\GG_i$ is \ac{CPT}. 
This can be seen by observing that the Choi matrix of $\mc T^{-1} \GG_i \mc T$ is given by 
$(U^*\otimes U) \mathrm{Choi}(\GG_i)^T(U^T \otimes U^{\dagger})$, which is positive definite for $\GG_i$ being \ac{CPT}. 

However, actual gate set implementations are noisy and hence 
not 
universal in the sense that they 
cannot
prepare any pure state. 
Therefore, in practice, the gauge freedom can be larger \cite{Nielsen2020GateSetTomography}.
  For instance, if all gate implementations are given by unital channels then an additional freedom exists: 
  depolarizing noise can be commuted through the circuit. 
  Therefore, it can be distributed arbitrarily among initial state, gates, and measurement. 
Meaningful distance measures for gate sets should have the same gauge freedom as the \ac{GST} data. 
The problem of finding gauge invariant distance has been studied by Lin et al.\ \cite{Lin19OnTheFreedom}. 
For individual gate sequences, any measure that compares only the ideal and observed outcome probabilities is naturally gauge invariant. The authors thus propose to use the total variation error, a natural error measure to compare probability distributions, for individual gate sequences. 
Let 
\begin{equation}\label{eq:output_probabilities}
  p_{j|\vec i}(E_j,\GG_{\vec{i}},\rho) \coloneqq \Tr[E_j\GG_{\vec{i}}(\rho)]
\end{equation}
denote the probabilities of measuring the j$th$ output of the \ac{POVM} with elements $E_j$ after applying the sequence $\vec i$ of gates in $\GG_{\vec{i}}$ to the state $\rho$. 
The \emph{total variation error} for sequence $\vec i$ between two gate sets 
$\hat{\mathcal{X}} = \left\{(\hat{E}_j),\hat{\GG},\hat \rho\right\}$ and $\mathcal{X}=((E_j),\GG,\rho)$ is defined as 
\begin{align}
\delta d_{\vec{i}}(\hat{\mathcal{X}},\mathcal{X})
\coloneqq
\frac{1}{2} \sum_{j}\left|\Tr\myleft[\hat{E}_j^{\dagger} \hat{\GG}_{\vec i}(\hat{\rho})\myright]-\Tr\myleft[E_j^{\dagger} \GG_{\vec i}(\rho)\myright]\right| \, .
\end{align}
The \emph{\ac{MVE}} is defined as \cite{Lin19OnTheFreedom}
\begin{align} \label{eq:MVE def}
\operatorname{MVE}_I(\hat{\mathcal{X}}, \mathcal{X}) 
\coloneqq
\EE_{\vec i \sim I}\bigl[ \delta d_{\vec{i}}(\hat{\mathcal{X}},\mathcal{X}) \bigr] 
\end{align}
w.r.t.\ a set of sequences $I$,
where $\vec i \sim I$ means that $\vec i$ is drawn uniformly from $I$. 
Often, we omit the subscript $I$ in the following. 
The \ac{MVE} corresponds to taking the natural worst case error measure over the measurement outcomes (the total variation distance) and averaging it over the available gate sequences. 
Often $I$ is chosen as the set of all gate sequences up to some length $\seqlength$. 
Then the expectation value \eqref{eq:MVE def} contains a sum over exponentially many terms. 
However, since they are all non-negative, they can be estimated sampling efficiently via Monte Carlo sampling \cite{Lin19OnTheFreedom}.

A closely related error measure is the \emph{\ac{MSE}}
\begin{equation} \label{eq:MSE def}
\mathcal{L}_{I}(\hat{\mathcal{X}},\mathcal{X})\coloneqq \EE_{\vec i \sim I} \sum_{j \in [n_E]} \left(\Tr\myleft[\hat{E}_j^{\dagger} \hat{\GG}_{\vec i}(\hat{\rho})\myright]-\Tr\myleft[E_j^{\dagger} \GG_{\vec i}(\rho)\myright]\right)^2 \, ,
\end{equation}
which averages the squared deviation over all sequences and POVM elements. 

\section{GST data processing via Riemannian optimization} \label{sec:data processing via Riemannian optimization}

In the previous section, we defined the compressive \ac{GST} data processing problem and introduced metrics for the quality of reconstruction.
We now turn to devising a concrete algorithm for the data processing problem. 
To this end, we formulate the reconstruction problem as a constraint optimization problem of a loss-function for the data fitting. 
A natural candidate for the loss-function is the \ac{MVE} restricted to the set of measured sequences. 
As a proxy we instead minimize the \ac{MSE} which depends smoothly on the gate set and is, therefore, more suitable for local optimization. 
In terms of the compressive device description, the \ac{MSE} \eqref{eq:MSE def} can be written as
\begin{equation} \label{eq:objective function}
\mathcal{L}_{I}(
A
, \K, B| \vec{y})
\coloneqq \frac{1}{|I|} \sum_{\vec i \in I} \sum_j \left(p_{j|\vec i}(
A
,\K,B)-y_{j|\vec i}\right)^{2} 
\end{equation}
where $y_{j|\vec i}$ is the empirical estimate of $\Tr[E_j \mc G_{\vec i}(\rho)]$.
Correspondingly, the compressive \ac{GST} data processing problem can be cast as the constraint optimization problem:
\begin{equation} \label{eq: estimation problem}
\begin{aligned}
\operatorname*{minimize}_{A, \mc K, B}
\quad &\mathcal{L}_{I}(A,\K,B | \vec y)
\\
\text{subject to} \quad  
&  \sum_{l=1}^{r_K} \K_{il}^{\dagger} \K_{il} = \mathds{1} \qquad \forall i \in [n], 
\\
& \sum_{j= 1}^{r_E} A_j^{\dagger} A_j = \mathds{1}\, ,
\\
&\fnorm{B}=1\, .
\end{aligned}
\end{equation}

The constraints restrict the objective variables to embedded matrix manifolds.  
Therefore, algorithms for the  optimization problem can be derived by generalizing standard optimization algorithms for functions on the Euclidean space to the geometric structure of these manifolds. 
 
\subsection{The complex Stiefel manifold} \label{sec: manifold theory}

In order to formulate our main reconstruction algorithm we need to understand the matrix manifold that encompasses the physicality constraints mentioned in Section \ref{sec:phys. constr.}. We start by summarizing the elementary properties of these manifolds, to then derive a parametrization of geodesics and the Riemannian Hessian, thereby extending what was previously done for their real counterparts in Ref.~\cite{edelman1998geometry}. For a comprehensive introduction to optimization on matrix manifolds we refer to the book by Absil, Mahony and Sepulchre \cite{Absil09}. 

Let $(K_l)_{l\in[r]}$ be the Kraus operators of a fixed gate. By stacking them along their row dimension to a new matrix $K \in \mathbb{C}^{dr \times d}$, we can write the \ac{CPT} constraint as $K^{\dagger} K = \mathds{1}$. 
In the following we set $D = dr$. 
The set 
\begin{equation}
  \St(D,d) \coloneqq \{K \in \mathbb{C}^{D \times d}: K^{\dagger}K=\mathds{1}_{d} \}
\end{equation}
is called the \emph{$D \times d$ complex Stiefel manifold}. 
This manifold is the set of isometries of the Euclidean space and contains the special cases of  the sphere $\St(D,1)$ and the unitary matrices $U(D) = \St(D,D)$. 
We regard it here as a submanifold of $\mathbb{C}^{dr \times d}$. 

The tangent space of $\St(D,d)$ at $K$ is given by 
\begin{equation}
  T_{K} \St(D,d) = \{\Delta \in \mathbb{C}^{D \times d}:\ K^{\dagger}\Delta = - \Delta^{\dagger}K\}\, .
\end{equation}
The canonical inner product of $\Delta_1,\Delta_2 \in T_{K} \St(D,d)$ can be defined as
\begin{align} \label{eq:canonical metric def}
\langle \Delta_1,\Delta_2\rangle_{K} = \Re\left\{\Tr(\Delta_1^\dagger \Gamma \Delta_2)\right\}
\end{align}
with $\Gamma = \mathds{1}-\frac{1}{2}KK^{\dagger}$. 
Another choice is the standard Hilbert-Schmidt inner product of the embedding matrix space.  
However, the advantage of the canonical inner product is that it weights all degrees of freedom on the tangent space equally.
The Stiefel $\St(D,d)$ together with the metric given by \eqref{eq:canonical metric def} is a Riemannian manifold. 
The \emph{normal space} is defined by
\begin{equation*}
\begin{split}
  N_K \St(D,d)
  = 
  \bigl\{\Delta_\perp \in \mathbb{C}^{D \times d}: \langle\Delta,&\Delta_{\perp}\rangle_K = 0 \\ 
  &\forall \Delta \in T_K\St(D,d)\bigr\}\, .
  \end{split}
\end{equation*}
The projector onto the normal space at position $K$ is given by 
\begin{equation}
  P_{N} (X) = K (K^{\dagger}X + X^{\dagger}K)/2
\end{equation}
for $X\in \CC^{dr \times d}$
and we can write the projector onto the tangent space at $K$ as 
\begin{equation}
  P_{T} (X) = X - P_{N}(X)\, .
\end{equation}

We wish to optimize the \ac{MSE} over $\St(D,d)$. 
In analogy to the optimization over $\operatorname{U}(n)$ in \cite{Abrudan2009ConjugateGradient}, we will move along geodesics, which are the locally length minimizing curves. 
In Appendix \ref{app:geodesic} we show that within $\St(D,d)$, a geodesic starting at $K_{t=0} \equiv K$ and going in the direction $\Delta \in T_{K} \St(D,d)$ can be written as 
\begin{equation} \label{eq:geodesic def}
K_t(K,\Delta) =
\begin{pmatrix} K & Q\end{pmatrix} \exp\Bigl[ t \begin{pmatrix} K^{\dagger} \Delta & -R^{\dagger} \\ R & 0\end{pmatrix} \Bigr]
\begin{pmatrix}\mathds{1} \\ 0\end{pmatrix} ,
\end{equation}
with 
$Q$,$R$ given by the QR decomposition of $(\mathds{1}-KK^{\dagger})\Delta$.
Note that $\dot{K}_t |_{t=0} = \Delta$. Often simpler curves that just satisfy $K_0 = K$ and $\dot{K}_t |_{t=0} = \Delta$ are used instead of the geodesic in order to save computation time \cite{Absil09}. 
However, computing the exponential of the $2d$-dimensional matrix in Eq. \eqref{eq:geodesic def} provides no bottleneck in our scenario as the inversion of the $2n d^2 r_K$-dimensional Hessian is more costly (see Section \ref{sec:runtime}).

In order to identify the Riemannian gradient and Hessian, we generalize results from the real case \cite{edelman1998geometry} to the complex case. 
Then we use the second order Taylor approximation of the objective function, 
which will be given below in terms of the \ac{MSE} \eqref{eq:MSE def} along geodesics (see Appendix \ref{app: complex riem newton}). 
The same treatment can be applied to the \ac{POVM} given by the matrices $A_j$ from the decomposition~\eqref{eq:def:AB}, where we define $A$ as the matrix obtained from stacking the $A_j$ along their row dimension. 
The physicality constraint on $A$ is then equivalent to $A \in \St(dn_E, r_E)$ with $n_E$ being the number of \ac{POVM} elements and $r_E$ their maximal rank. 
Finally, the constraint $\fnorm{B}^2=\vvec(B)^{\dagger}\vvec(B)=1$ on the initial state~\eqref{eq:def:AB} can also be captured by the Stiefel manifold via the requirement $\vvec(B)\in \St(d \, r_\rho,1)$, where $\vvec(B)\in \CC^{d\,r_\rho}$ is the vectorization of $B\in \CC^{d\times r_\rho}$. 

\subsection{The \texorpdfstring{\alg}{mGST} estimation algorithm}
\label{sec:estimation}
With a better understanding of the underlying manifold structure we can now formulate a concrete optimization approach to tackle the estimation problem \eqref{eq: estimation problem}. 
The least squares cost function \eqref{eq:MSE def} is a polynomial of order at most the sequence length squared in the parameters of $\GG$, with a highly degenerate global minimum due to the gauge freedom. 
In analogy to the alternating minimization techniques which are successful for matrix product state completion \cite{SuessPhDThesis,Grasedyck_2015_Variants,2016arXiv160905587W} we alternate between updates on $A,\K$ and $B$. Each update would naively be done via a local optimization approach such as gradient descent. However, we observe that following the gradient direction on the respective manifolds is problematic around saddle points, which are frequently encountered in our optimization problem. 
In principle the gradient direction points away from saddle points, yet the norm of the gradient can be arbitrarily small. There are different approaches in the literature to deal with this problem. For instance information about the curvature can be included \cite{Dauphin2014IdentifyingAndAttacking} or, if a saddle point is encountered, random update directions can be chosen to escape the area of vanishing gradient \cite{sun2019escaping,jin2021nonconvex}.
We find that the so-called \ac{SFN} method \cite{Dauphin2014IdentifyingAndAttacking} yields considerably better results than first order methods. 
There the update direction is given by $-\abs{H}^{-1}g$ with $H$ being the Hessian and $g$ the gradient and the absolute value $\abs{H}$ define by spectral calculus. 
An instructive way to see why this leads to a speedup is to write the Hessian $H$ as $H = \sum_i \lambda_i \ketbra{v_i}{v_i}$, where $v_i$ is the eigenvector to eigenvalue $\lambda_i$. 
The update direction of the \ac{SFN} method then reads $-|H|^{-1}g = -\sum_i |\lambda_i|^{-1} \ket{v_i}\braket{v_i}{g}$. 
Since the vectors $\ket{v_i}$ form a basis, this can be interpreted as a rescaling of $\ket g$ by $|\lambda_i|^{-1}$ in the directions $\ket {v_i}$. 
As with the standard Newton method, this leads to a large rescaling if the curvature in a particular direction is small, resulting in large steps even close to the saddle point. Taking the absolute value of the eigenvalues then ensures that saddle points are repulsive. 
For numerical stability it is beneficial to introduce a damping term that offsets the eigenvalues of $H$ that are very close to zero before the inversion. 

Algorithm~\ref{alg:SFN-update} describes a single step of the damped saddle-free Newton method with damping parameter $\lambda$ and is a generalization of the original \ac{SFN} method \cite{Dauphin2014IdentifyingAndAttacking} to manifolds.

\begin{algorithm}
\SetKwInOut{Input}{input}{}{}
\Input{Curve parametrization $Y_t(Y_0,\Delta)$, objective function $\mathcal{L}_I(Y_t)$, damping parameter $\lambda$}
Compute the gradient $\vec{G}$ and Hessian $H$ of $\mathcal{L}_I(Y_t)$ at~$Y_0$.
\\ 
Determine the update direction 
  \begin{align*}
    \begin{pmatrix} \vec{\Delta} \\  \vec{\Delta}^{*} \end{pmatrix} = 
    \left( |H| + \lambda \mathds{1}\right)^{-1} \begin{pmatrix} \,\, \vec{G}^* \\  \vec{G} \end{pmatrix} .
  \end{align*} \\
Determine the step size $\tau = \underset{t}{\mathrm{argmin}} \, \mathcal{L}_I\left(Y_t(Y_0,\vec{\Delta})\right)$. 
\KwRet{$Y_{\tau}(Y_0,\vec{\Delta})$}
\caption{SFN update} \label{alg:SFN-update}
\end{algorithm}

Algorithm~\ref{alg:SFN-update} is formulated in a way that is compatible with an update in Euclidean space as well as an update on the Stiefel manifold. 
In Euclidean space we update along the curve $Y_t(\argdot,\argdot): \mathbb{C}^{D \times d} \times \mathbb{C}^{D \times d} \rightarrow \mathbb{C}^{D\times d}$ with $Y_t(Y_0, \Delta) = Y_0 + t\, \Delta$ for an update direction $\Delta$. 
On the Stiefel manifold we have $Y_t(\argdot, \argdot): \operatorname{St}(D,d) \times \mathfrak{T} \rightarrow \operatorname{St}(D,d)$ with the curve given by the geodesic~\eqref{eq:geodesic def} and $\mathfrak{T}$ being the tangent bundle on $\operatorname{St}(D,d)$. 
The step size is determined by locally optimizing over the parameter $t$ using standard gradient free optimizers.
We derive an expression for the Hessian on the Stiefel manifold in Appendix~\ref{app: complex riem newton}. 
In Appendix~\ref{app:Euclidean gradient and Hessian}, we also provide a detailed discussion and expressions for the optimization in complex Euclidean space.

\begin{algorithm}
  \SetKwInOut{Input}{input}
  \Input{Data $\{\vec y_{j|\vec i}\}_{\vec i \in I,\ j \in [n_E]}$, batch size $\kappa$, Kraus rank $r_K$, initialization $(A^0, \K^0, B^0)$, stopping criterion}
  $i \leftarrow 0$ \\
  \Repeat{stopping criterion is met at $i = i_\ast$}{%
    Select batch $J \subset I$ of size $|J| = \kappa$ at random \\
    $A^{i+1} \leftarrow$ update $A^{i}$ with objective $\mathcal{L}_J(\argdot,K^i, B^i; \vec y)$ along geodesic on $\operatorname{St}(dn_E,d)$ \\
    $\K^{i+1} \leftarrow$ update $\K^{i}$ with objective $\mathcal{L}_J(A^{i+1}, \argdot, B^i; \vec y)$ along geodesic on $\operatorname{St}(r_Kd,d)^{\times \ngates}$ \\
    $B^{i+1} \leftarrow$ update $B^{i}$ with objective $\mathcal{L}_J(A^{i+1}, \K^{i+1}, \argdot; \vec y)$ along geodesic on $\operatorname{St}(d^2,1)$ \\ 
    $i \leftarrow i + 1$     
  }
  \Return $(A^{i_\ast}, \K^{i_\ast}, B^{i_\ast})$ 
  \caption{mGST\label{Alg:main}}
\end{algorithm}

Algorithm~\ref{Alg:main} describes the main \alg\ routine. 
It can be run with different choices of smooth objective functions, and we use the \ac{MSE} \eqref{eq:MSE def} by default. In our numerics we often find that optimizing the log-likelihood function after the \ac{MSE} can improve estimates (see Appendix \ref{app:MLE} for a discussion).

The algorithm alternates updates on $A$, $\K$ and $B$. Updates are performed using Algorithm~\ref{alg:SFN-update}
on the tangent spaces of the respective Stiefel manifolds.
In order to achieve good convergence, we run the optimization with \alg\ in two consecutive steps: 
we start from a random initialization and perform a coarse grained optimization with a small batch size $\kappa$, i.e.\ only using $\kappa$ many random gate sequences from $I$ for each update step. 
The batching of data results in lower computation time for the derivatives and adds a factor of randomness to the optimization, which avoids getting stuck at suboptimal points to a certain degree. 
We terminate the first optimization loop when the objective function $\mathcal{L}_{I}(A^{i},\K^{i},B^{i}|y)$ is smaller than an \emph{early stopping value} $\delta$, which is obtained from the data as follows. 

For a number of $\nsamples$ samples per sequence the outcome probabilities of each sequence for the true gate set are given by 
\begin{equation}\label{eq:y_and_k}
  y_{j|\vec{i}} = k_{j|\vec{i}}/\nsamples\, ,
\end{equation}
where $k_{j|\vec{i}}$ is the number of times outcome $j$ is measured upon applying the gate sequence $\vec i$. 
Due to Born's rule, $k_{j|\vec{i}}$ is
distributed according to the multinomial distribution $\mathrm{M}(\nsamples,(p_{1|\vec i}, \dots, p_{n_E|\vec i}))$ with probabilities $\{p_{j|\vec i}\}_j$ and $\nsamples$ trials. 
We estimate the expectation value of the objective function from the values $y_{j|\vec{i}}$.  
This provides us with a rough estimate for how low the objective function value can become, given the sample counts $k_{j|\vec{i}}$.  
Then we set the early stopping value to be twice that estimate, 
\begin{equation}
  \delta \coloneqq 2\, \EE_{\tilde{k}_{j|\vec{i}} \sim \mathrm{M}(\nsamples, (y_{j|\vec i}))}\frac{1}{|I|} \sum_{\vec i \in I} \sum_j \left(y^j_{\vec i}-\tilde{k}_{j|\vec i}/\nsamples\right)^{2} .
\end{equation}
Hence, we require the objective function on the full data set to be close to its expectation value for the measured probabilities $y_{j|\vec i}$ obtained from $\nsamples$ samples. 

While computationally inexpensive the mini-batch stochastic optimization does not converge
to an optimal point on the full data set $I$. 
In a second optimization loop, we initialize the \alg\ algorithm 
with the result from the first run and use all the data for the updates. 
Formally, we choose the batch size $\kappa = |I|$ and, thereby, make the random batch selection obsolete. 
We perform these more costly update steps until the change in objective function reaches a desired relative precision $\epsilon$, 
\begin{equation} \label{eq: relative precision epsilon}
  \mathcal{L}_{I}(A^{i},\K^{i},B^{i}|\vec y) - \mathcal{L}_{I}(A^{i-1},\K^{i-1},B^{i-1}|\vec y) \leq \delta\epsilon\, 
\end{equation}
or a maximal number of iterations is exceeded.

The first optimization run is initialized with a random gate set parameterized by $A^0,\K^0$ and $B^0$ (see Section~\ref{sec:phys. constr.}). 
For the random initialization we make use of the 
\emph{\ac{GUE}}. 
A matrix $H$ belongs to the \ac{GUE} if $H = (M + M^{\dagger})/2$, where 
$M$ is a \emph{complex Gaussian matrix}, i.e, real and imaginary part of each $M_{ij}$ are independently drawn from $\mathcal{N}(0,1)$, the normal distribution with zero mean and unit variance. 
In this case we write $H \sim \mathrm{GUE}$. 
For $A^0$ and each gate in $\K^0$ we take the first $d$ columns of $\e^{\i H}$ with $H \sim \mathrm{GUE}$ to obtain a random isometry $\K^0$. 
For $B_0$ we take a complex Gaussian matrix and normalize it such that 
$\Tr[B^{0\dagger}B^0] = 1$. 

Importantly, due to the nature of non-convex optimization, several initializations can be needed to converge to a satisfactory minimum.

\section{Numerical analysis} \label{sec:Numerical analysis}

\begin{figure*}[t] 
    \centering
    \includegraphics{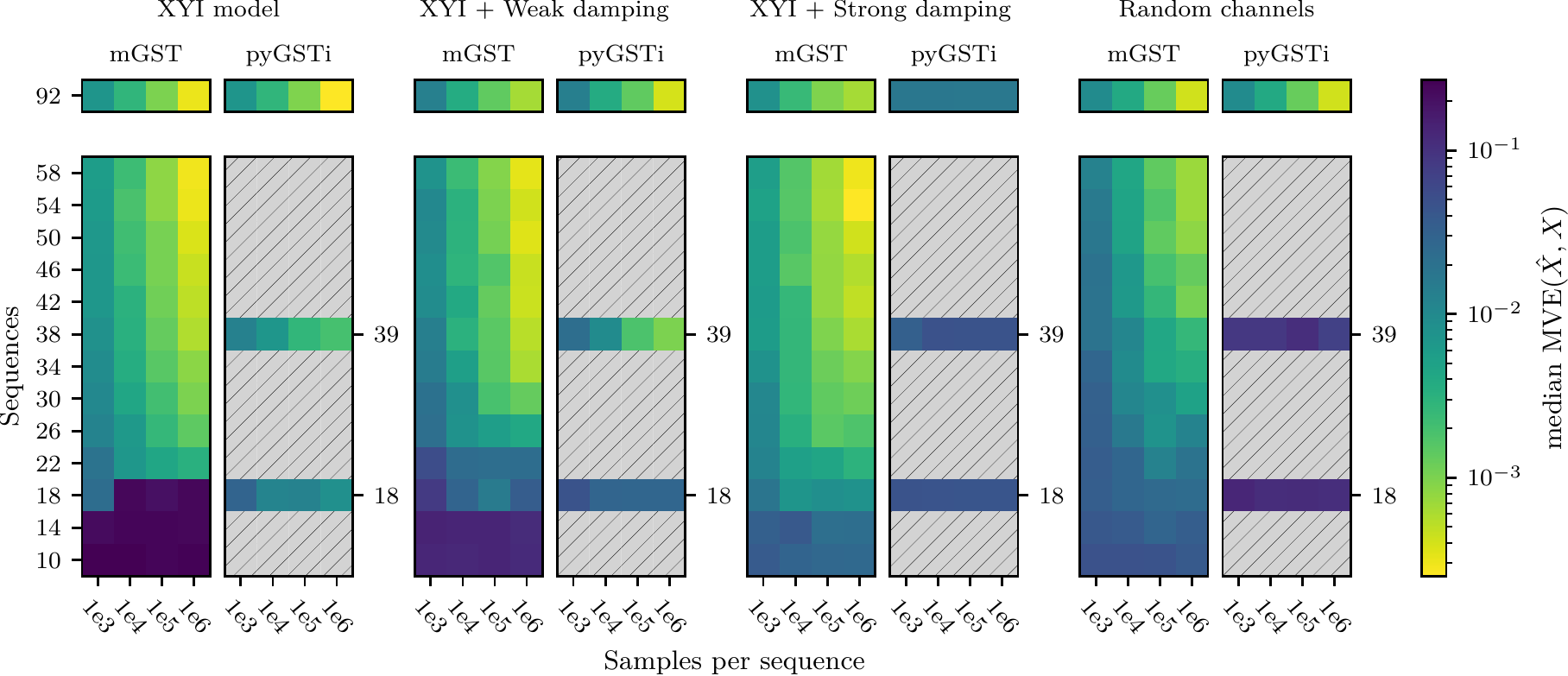}
  \caption{\Acf{MVE} comparison between \alg\ (with log-likelihood cost function and $r_K = 4$) and \pyGSTi\, showing the dependence on the number of sequences, the number of samples per sequence and the gate set on a single qubit. 
  The number of sequences used by \alg\  in the range 10-58 are drawn uniformly at random, while the sequences for \pyGSTi\ need to follow the \pyGSTi\ fiducial design and are limited to the fixed sequence counts (18 and 39). We choose independently drawn random fiducials for each instance.
  The 92 sequences used by both \alg and \pyGSTi\ are taken from the standard \pyGSTi\ sequence design for the XYI-model. All sequence are of length $\seqlength=7$.
  The \ac{MVE} depicted in each square is the median result for 10 instances, each with random statistical measurement noise and a random sequence drawn from the uniform distribution. 
  In the random channel scenario, a new random channel is used for each instance. The XYI-model is a simple unitary model used in the GST literature and the weak damping model consists of amplitude damping noise on each gate with $\Gamma = 0.94$, while the strong damping model uses $\Gamma = -0.6$. A complete description of the models used can be found in the main text.
  Each model has additional depolarizing noise of strength $p = 0.01$ on the initial state.
  }
  \label{fig: Phase transitions - Amplitude damping}
\end{figure*}

In this section, we evaluate the performance of \alg\ in different scenarios in numerical simulations. 
In particular, we compare its performance to the state-of-the-art implementation for gate set tomography, \pyGSTi~\cite{Nielsen2020GateSetTomography}, in the regimes where both methods can be applied. 

For \pyGSTi\ to be applicable one has to use structured gate sequences inspired by standard quantum process tomography. 
In Section \ref{sec:Gate set structure} we evaluate the performance of \alg\ and \pyGSTi\ on minimal measurement sequences and different models to find that \alg\ benefits from flexibility in the sequence design and a fully general model parametrization.
Section \ref{sec: Sample size dependence} numerically validates the expected inverse square-root scaling of the reconstruction error with the number of measurement samples per sequence for different noise regimes. 
Section \ref{sec:Sequence count dependence} numerically determines the required number of random sequences to accurately reconstruct simple and random gate set models with \alg\, for different Kraus ranks. In Section \ref{sec: 3Q unitary} we follow up with a numerical demonstration of unitary noise characterization for a three-qubit gate set using a priori knowledge in the initialization. 
Finally, in Sections \ref{sec:Calibration} and \ref{sec:runtime} we discuss the choice of initialization and hyperparameters, as well as the runtime of \alg.

For a model of $n$ gates reconstructed from $m$ measurements of sequence length $\seqlength$, we validate the performance of \alg\ by computing the \ac{MVE}~\eqref{eq:MVE def} over all possible $n^\seqlength$ sequences, or $10^4$ random sequences of length $\seqlength$ if $n^\seqlength > 10^4$. 
Usually $m \ll \mathrm{min}(n^\seqlength, 10^4)$ and the \ac{MVE} can be thought of as a generalization error on the predicted output probabilities of the gate set estimate. 
The gate sets studied in this section all use the same target initial state $\ketbra{0}{0}$ and computational basis measurement, although with different levels of noise applied to them. 
For instance, we often use global depolarizing noise, which acts on a quantum state $\rho$ as $\rho \mapsto (1-p) \rho + p\, \mathds{1}/d$. 
For the numerics presented here, we use a maximum of 100 reinitializations (if not stated otherwise). 
A discussion of the required number of initializations is given in Section \ref{sec:Calibration}.
A Python implementation of \alg\ and a short tutorial can be found on GitHub \cite{mGST2021}.

\subsection{Gate set and measurement structure} \label{sec:Gate set structure}
  We compare \alg\ and \pyGSTi\ for the minimal number of sequences doable with each method and for gate sets of different conditionings, without using the compression capabilities of \alg\ yet. 
  We find that \alg\ is more flexible in the sequence design and model parametrization, while generating estimators with lower mean variation errors in several regimes.
 
  The traditional strategy for \ac{GST}, akin to standard quantum process tomography, is to generate a frame for $\L(\mathcal{H})$, measure each gate in that frame and generate an estimate for each gate by applying the pseudo-inverse of the measurement operator.

  This is particularly important for the first reconstruction step in \pyGSTi\, where the sequences that generate the frame are called \emph{fiducials}. 
  The strategy of \pyGSTi\, is to obtain an initial estimate via the pseudo-inverse, followed up by local optimization of a particular cost function \cite{Nielsen2020GateSetTomography}. In contrast, we perform \alg\ using random initializations and, thereby, not rely on designated fiducial sequences.

  Figure~\ref{fig: Phase transitions - Amplitude damping} compares \alg\ to \pyGSTi\, focusing on the regime of very few 
  gate
  sequences, showing what is needed in terms measurement effort to obtain low mean variation errors for different gate sets. 
  In order to test \pyGSTi\ in the regime of low sequence counts, we replace the 5 standard fiducial sequences with  2 or 3 fiducial sequences drawn uniformly at random, thereby reducing the total sequence number from 92 to 18 or 39 sequences. Since \alg\, is compatible with any sequences design, we use between 10 and 58 random sequences for \alg\ to explore the low sequence count region. 

  The first gate set we study is the so-called \emph{XYI model}, the standard single qubit example in the \pyGSTi\ package \cite{Nielsen20pyGSTi}. 
  The XYI model consists of the identity gate, a $\pi/2$ X-rotation and a $\pi/2$ Y-rotation on the Bloch sphere, with initial state $\ketbra{0}{0}$ and measurement in the computational basis. 
  Results for the ideal XYI-model can be seen on the left in Figure~\ref{fig: Phase transitions - Amplitude damping}, with \alg\ and \pyGSTi\ performing identically for 92 sequences. Comparing the results for 18 sequences we find that \alg\ does not converge on more than 
  50\% of trials, which reflects in the median MVE being above $10^{-1}$, while pygsti achieves lower median errors. Comparing the 38 and 39 sequence medians however, we find that \alg\ yields lower error models than \pyGSTi. 

  The subsequent models analyzed successively deviate from the simple unitary XYI-model and highlight the versatility of 
  our
  manifold approach. 
  Since the full \ac{CPT} parametrization used in our optimization \eqref{eq: estimation problem} is agnostic to any special gate set properties we expect it to perform well for all possible \ac{CPT} maps as gate implementations. 
  For instance, for random and specific non-Markovian channels. 
  \pyGSTi\ on the other hand uses a parametrization of Lindblad type and is therefore based on a more limited model space. 

  To illustrate this comparison, 
  we perturb the XYI-model by adding amplitude damping noise to each gate. The amplitude damping channel can be written in terms of the Kraus operators $K_1 = \begin{pmatrix} 1 & 0 \\ 0 & \Gamma \end{pmatrix}$ and $K_2 = \begin{pmatrix} 0 & \sqrt{1-|\Gamma|^2} \\ 0 & 0\end{pmatrix}$, which arise e.g.\ from the Jaynes-Cummings model of a qubit system interacting with a quantized bosonic field \cite{2020npjQI...6....1G}. 
  How well \alg\ and \pyGSTi\ perform on a model with $\Gamma = 0.94$ can be seen in the center left block of Figure~\ref{fig: Phase transitions - Amplitude damping} (XYI + Weak damping). 
  We find generally similar performance, with \alg\ being a bit more accurate on 38 sequences, and a bit less accurate on 92 sequences with $10^6$ samples. 

  Increasing the interaction time between qubit and environment leads to memory effects and strong non-Markovianity of the amplitude damping channel at $\Gamma = -0.6$. This scenario is shown in the center right plot of Figure~\ref{fig: Phase transitions - Amplitude damping}, and we see that while the accuracy of \alg\ is the same as before, the model parametrization of \pyGSTi\ cannot fit the model with \acp{MVE} below $10^{-2}$, independent of the sequence or sample count. 

  For the last comparison (rightmost block in Figure~\ref{fig: Phase transitions - Amplitude damping}) we look at the performance for random full Kraus rank channels. 
  Each channel is constructed by drawing a Haar random $d^3 \times d^3$ unitary and then taking its first $d$ columns. The resulting $d^3 \times d$ matrix is an isometry and therefore constitutes a valid set of Kraus operators. Note that this construction is different from the previous construction of random channels via the Gaussian unitary ensemble. 
  The results show that \alg\ can reconstruct these models from low sequence counts, while \pyGSTi\ does not yield good estimators. Using the standard sequence design of 92 sequences, \alg\ and \pyGSTi\ have identical accuracy again, suggesting that random channels are typically well within the model space of \pyGSTi\ after all.
  These demonstrations show that \alg\ is indeed flexible in the sequence design with state-of-the-art performance for arbitrary gate set implementations.

\subsection{Number of samples per sequence} \label{sec: Sample size dependence}

\begin{figure}[t]
  \includegraphics{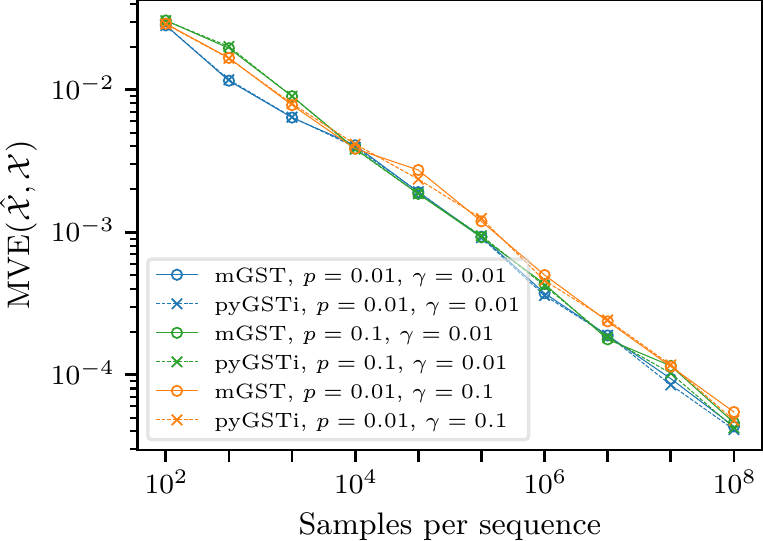}
  \caption[]{Reconstruction of the XYI gate set for different levels of depolarizing noise with strength $p$ 
  and unitary noise with strength $\gamma$ on each gate. 
  The unitary noise is given by $\e^{\i\gamma H}$ with $H \sim \mathrm{GUE}$. Additional depolarizing noise with $p = 0.01$ is applied before measurement. The \alg-algorithm is run on the log-likelihood cost function with $r_K = 4$ (max.), which is the same as for \pyGSTi. 
  As gate sequences we used again the standard \pyGSTi\ fiducial sequences, with the number of measurements per sequence between $10^2$ and $10^8$. 
  The lines connect data points of which each is the median over $10$ runs. 
  For each run a new random overrotation is drawn and new measurements are simulated. 
  The measurement sequences are the $92$ sequences provided by the \pyGSTi\ software, with a maximum sequence length of $\seqlength \leq 7$.
  }
  \label{fig:shot noise}
\end{figure}
The probability associated to every sequence is estimated from a finite number of samples. 
Here, we study the resulting effect on the reconstruction accuracy as measured by the \ac{MVE} more closely. 

For a high number $m$ of samples per sequence, each probability $y_{\vec{i}}^j$ in the objective function is estimated with an error of order $1/\sqrt{m}$.
Therefore, we expect the \ac{MVE} to also decrease as $1/\sqrt{m}$ if the algorithm converges to the global minimum. 
This scaling was observed to hold true for \pyGSTi~\cite{Nielsen2020GateSetTomography}. 
In order to be able to compare the scaling of \alg\ directly to the one of \pyGSTi, we use a standard \pyGSTi\ setting: 
The gate set is the XYI-model (with $\pi/2$-rotations) and the gate sequences are the standard \pyGSTi\ sequences for this model with a maximum sequence length of $\seqlength=7$. 

We add noise to the gate set by varying the amount of depolarizing noise with strength $p$ on each gate and also overrotating each gate by a random unitary. 
The random unitaries are given by $\e^{\i\gamma H}$ with $H \sim \mathrm{GUE}$. 
In particular, this means that $H$ can be bounded on average as follows.
We can write $H = (M_1 + M_1^T + \i(M_2 - M_2^T))/2$ with $M_i$ being independent Gaussian matrices. 
Next, we use Gordon's theorem for Gaussian matrices (see e.g.\ \cite[Theorem 5.32]{Ver12}), which tells us that $\EE \pnorm[\infty]{M_1} \leq 2\sqrt d$. 
The relevant magnitude of the random generator $H$ is then in expectation upper bounded as 
\begin{equation}
  \EE \norm{H}_{\infty} \leq 2\, \EE \norm{M_1}_{\infty} \leq 4 \sqrt{d}\, .
\end{equation}
State preparation and measurement are assumed to be noise-free in this setup, however for a fixed sequence length the depolarizing noise per gate is equivalent to a global depolarizing channel applied before measurement, since it commutes with the unitary gates. 
Figure~\ref{fig:shot noise} depicts the resulting \ac{MVE}-scaling of the reconstruction where data was generated using different numbers of samples per sequence~$m$. 

We observe that \alg\ follows the expected scaling in $m$, matching the scaling of \pyGSTi\ for different levels of unitary and depolarizing noise.

\begin{figure*}[t] 
    \begin{minipage}[b]{.4\paperwidth}
            \includegraphics{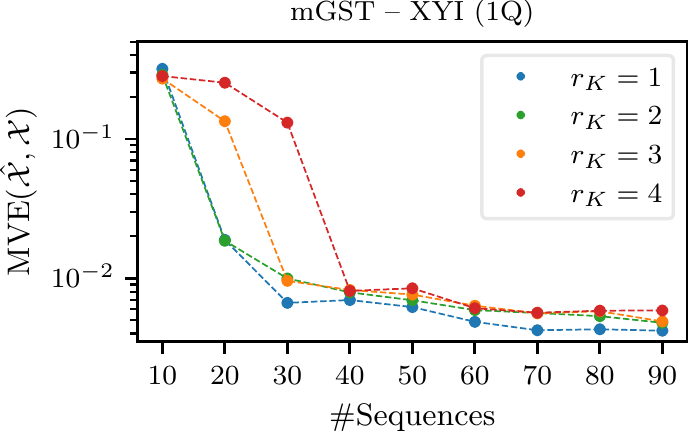} \\[.3cm]
            \includegraphics{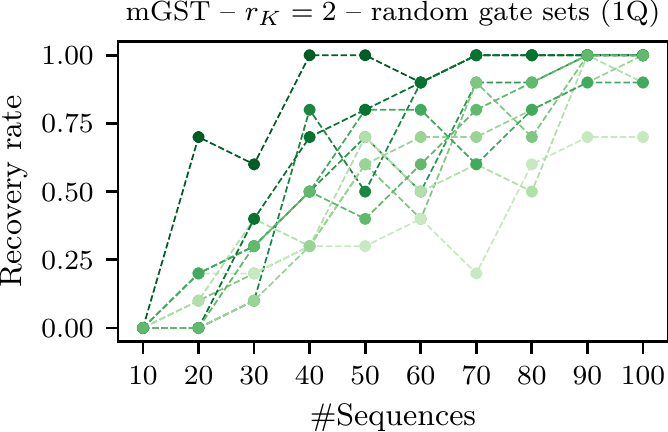}
    \end{minipage}  
    \begin{minipage}[b]{.4\paperwidth}
            \includegraphics{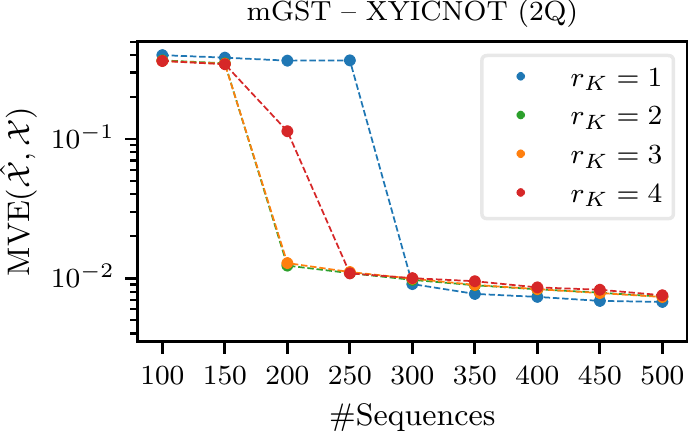} \\[.3cm]
            \includegraphics{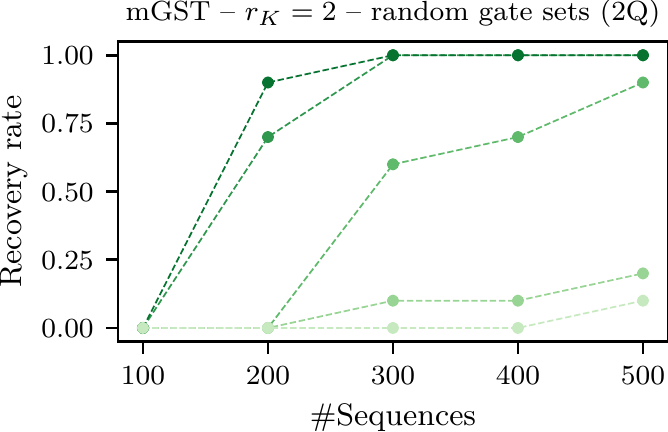}
    \end{minipage}
  \caption{\label{fig:Phase transitions}
  Median error of \alg\ run on the least squares cost function, plotted over the number of sequences for a single qubit model (\textbf{top left}) and a two-qubit model (\textbf{top right}). Each data point is the median over the results from $10$ different random sequences.  
  The measurement data for the XYI - and the XYICNOT gate set is taken from a noisy version with depolarizing noise of strength $p= 0.001$ on each gate, depolarizing noise with strength $p = 0.01$ before measurement, as well as independent random unitary rotations $\e^{\i\gamma H}$ with $\gamma = 0.001$ and $H \sim \mathrm{GUE}$ on each gate. 
  \newline
  On the \textbf{bottom left} the recovery rates for the reconstruction of
  different 
  models of $3$ Haar random unitaries are shown. 
  For each gate set the average over $10$ draws of random sequences is shown. 
  A gate set is classified as recovered if the \ac{MVE} falls below $0.03$. 
  The \textbf{bottom right} depicts the recovery rate for random two qubit 
  gate sets of the form $\mathcal{G} = \{\mathds{1}_4, \mathds{1}_2\otimes U_1, \mathds{1}_2 \otimes U_2, U_1 \otimes \mathds{1}_2, U_2\otimes \mathds{1}_2, U_{12}\}$ where $U_1$ and $U_2$ are Haar random single qubit unitaries and $U_{12}$ is a Haar random two qubit unitary.
  Additionally, each single qubit and two qubit gate
  contains depolarizing noise of strength $p = 0.001$ and depolarizing noise of strength $0.01$ is applied before measurement. 
  The recovery rate is averaged over $10$ random sequence draws. 
  For all 
  gate sets the 
  sequences are drawn uniformly at random with sequence length $\seqlength = 7$ and $m = 1000$ samples per sequences. 
  The maximum number of initializations are $80$, $33$, $17$ and $10$ for Kraus ranks $1, \dots,4$ respectively. They are chosen such that the maximal computation time is equal among different ranks. 
  } 
\end{figure*}

\subsection{Number of sequences} \label{sec:Sequence count dependence}

The arguably most challenging experimental requirement of \ac{GST} is the number of measurement settings (sequences) that are required for a successful gate reconstruction. 
One of the main motivations of compressive \ac{GST} is to employ structure constraints, 
i.e.\ to reduce
the number of degrees of freedom of the reconstruction problem, in order to reduce the required number of measurements. 
Instead of reconstructing arbitrary quantum channels we aim at reconstructing low-rank approximations of the gate set elements. 
In addition, we expect that by using the \alg\ algorithm, compressive recovery is possible from already a `few' randomly selected sequences. 
We here numerically demonstrate that this is indeed the case.

The top row of Figure~\ref{fig:Phase transitions} shows the median performance in \ac{MVE} against the number of randomly chosen sequences for 
different Kraus ranks. 
On the left are the results for the single qubit XYI model as defined in Section~\ref{sec:Gate set structure}. On the right are the results for the XYICNOT gate set that is based on the identity, CNOT and Pauli-X and -Y rotations on each qubit individually, with rotation angle $\pi/2$. 

We observe a phase transition in the \ac{MVE} that indicates a minimal number of sequences that are required for the successful reconstructions of the gate sets. 
As expected, constraining the reconstruction to a lower Kraus rank indeed reduces the amount of required sequences in the reconstruction in most cases. 

An intriguing exception is the $r_K = 1$ reconstruction of the two qubit gate set that exhibits the worst reconstruction performance compared to higher rank constraints. 
We suspect that this is due to the optimization problem being more dependent on the initialization for $r_K = 1$. 
In more general settings, it has been observed that the optimization over matrix-product states with fixed Kraus rank can be unstable and using rank-adaptive optimization techniques yield much better performance \cite{GrasedyckKraemer:2019:Stable,GoessmannEtAl:2020:Tensor}. 
This motivates to use a slightly higher rank in the optimization than the expected rank of an effective approximation of the gate set. 
In accordance with this intuition, we find that it is beneficial to constrain the optimization to $r_K = 2$ in order to achieve an accurate unit-rank approximation.
The same effect is also observed in the single qubit example when taking a detailed look at the number of required initializations (see section \ref{sec:Calibration})
, yet less pronounced. 
In the bottom row of Figure~\ref{fig:Phase transitions} we show the recovery rates for random unitary models, with the reconstruction now using a fixed Kraus rank of $r_K = 2$. Note that there are three sources of randomness present in the data, first the Haar-random unitary gates, then the random drawing of gate sequences and finally the random initialization of the algorithm. Each shade of green corresponds to one random gate set, and the recovery rate tells us how many of the 10 random sequence sets lead to a successful reconstruction, given a budget of 33 initializations. We find that the random single qubit gate sets all have similar recovery rates, with a successful reconstruction possible from $n_{\mathrm{seq}} = 20$ to $n_{\mathrm{seq}} = 30$ sequences, and a high rate of recovery at $n_{\mathrm{seq}} = 100$ sequences. In the two qubit case a different picture emerges, where two of the random gate sets show a high recovery rate at $n_{\mathrm{seq}} = 200$ sequences (akin to the XYICNOT-model), while the least favorable random gate set was only recoverable at $n_{\mathrm{seq}} = 500$ sequences. 
This shows that random gate sets for two qubits can have very different conditioning.

\subsubsection*{Sequence number comparison to \texorpdfstring{\pyGSTi}{pyGSTi}}
Comparing the number of random sequences needed for \alg\ and the number of sequences for \pyGSTi\ is not straightforward. 
The standard \pyGSTi\ data-processing pipelines crucially relies on 
specific, fixed sequence construction. 
For this reason \pyGSTi\ cannot be applied to the type of data that we use here. 
We can however compare the number of random sequences with the number of deterministic sequences that the standard implementation of \pyGSTi\ uses. 
For the single qubit XYI-model, the minimal number of sequences given in the \pyGSTi\ implementation is $n_{\mathrm{seq}} = 92$. This is significantly larger than the number of random sequences at which the phase transition of \alg\ in Figure~\ref{fig:Phase transitions} appears. However, the $n_{\mathrm{seq}} = 92$ sequences are overcomplete by design, and we find that \pyGSTi\ can also reconstruct the XYI model with $n_{\mathrm{seq}} = 48$ sequences. 
Yet we find the same sequence design not to be successful for three Haar-random single qubit gates, indicating that the choice of sequences is well-tailored to the XYI-model. 
The reduction in sequences becomes more pronounced for the two-qubit gate set studied in the top right of Figure~\ref{fig:Phase transitions}. 
For this gate set, the minimal number of gate sequences that \pyGSTi\ uses is $n_{\mathrm{seq}} = 907$, which is significantly larger than what \alg\ needs. 

\subsection{Characterizing unitary errors using prior knowledge} \label{sec: 3Q unitary}
\begin{figure}[h!]
        \includegraphics{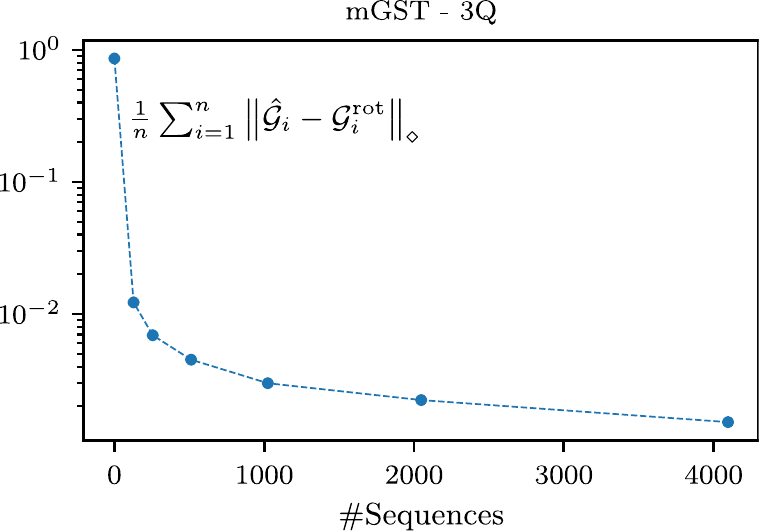}
  \caption[]{Average diamond distance between 3-qubit rotated target gates $\mathcal{G}^{\mathrm{rot}}_i$ and their unitary ($r_K = 1$) \alg\ estimators $\hat{\mathcal{G}}_i$ as a function of the number of sequences. 
  The 0-sequence data marks the average diamond distance between initialization $\mathcal{G}_i$ and $\mathcal{G}^{\mathrm{rot}}_i$. The rotated gates $\mathcal{G}^{\mathrm{rot}}_i$ are related to their counterparts $\mathcal{G}_i$ by independent random (global) overrotations on each gate, given by $\exp(\i \gamma H))$ with $H \sim \mathrm{GUE}$ and $\gamma = 0.05$, leading to $\EE\norm{\gamma H}_{\infty} \leq 2\sqrt{2}/5$.
  The data is simulated from the gates $\mathcal{G}^{\mathrm{rot}}_i$ with depolarizing noise of strength $p = 0.01$ on each gate and 
  the measurements are taken from random sequences of length $\seqlength = 7$, with $m = 10^5$ samples per sequence.
  }
  \label{fig:3Q_diamond}
\end{figure}

In the previous section we demonstrated compressive gate set tomography for one- and two-qubit gate sets using agnostic random initializations.  
A major obstacle in going beyond reconstructing entire two-qubits gate sets even compressively on desktop hardware is that besides run-time and storage also the number of required random initializations until proper convergence grows in principle with the number of qubits. This is due to longer sequences being required for tomographic completeness, leading to a higher order polynomial in the cost function. 
This situation can be remedied by using prior knowledge, such as the target gate set, for the initialization. 
In this case, a gate set in the vicinity of the initial point will be found which is in better agreement with the data. In a conceivable experimental scenario the gates are more or less known due to the physical setup and previous benchmarking rounds, but further calibration requires information about present coherent errors. The use of prior knowledge can be seen as a situational tool to further reduce runtime when applicable, but for general purpose verification and characterization no initial point is to be assumed for compressive \ac{GST} with the \alg\ algorithm. \\
To showcase the characterization of unitary errors using prior knowledge, we take a three-qubit gate set that is the direct generalization of the previous two-qubit XYICNOT model, by adding the local X- and Y-rotations as individual gates to the 3rd qubit and adding a CNOT between qubits 2 and 3. We then apply a global random rotation to each gate individually, as well as depolarizing noise on each gate. 
From random sequences of 
fixed sequence length we can then, in theory, fit the noisy model perfectly via an $r_K = 1$ approximation, as the depolarizing channels commute with the unitary gates and can be pulled into initial state or measurement. 
In Figure~\ref{fig:3Q_diamond} we see that \alg\ is indeed able to precisely reconstruct the rotated gates, as shown by the average diamond norm error. 
We chose a comparatively high number of $10^5$ samples per sequence to showcase that high accuracy can indeed be realized using this method: for instance, only $256$ sequences are enough to achieve an average diamond norm distance of around $0.007$ between the reconstructed unitary gates and the true unitary gates, 
which include overrotations. 
The fact that these overrotations were modelled as being global on all $3$ qubits suggests that we can efficiently characterize unitary crosstalk as well, by capturing the effect of single and two qubit gate on their neighbours within a three qubit region.

\subsection{Implementation details and calibration} \label{sec:Calibration}

We now provide more details on the simulations, the criteria for successful recovery and the required number of initializations.
To simulate measurements on a gate sequence $\vec{i}$, we first compute the outcome probabilities $p_{j|\vec{i}}$ from Eq.~\eqref{eq:output_probabilities} of the \ac{POVM} elements according to the model gate set in question. 
Afterwards we draw $\nsamples$ samples from the multinomial distribution $M(\nsamples,(p_{1|\vec{i}},\ \dots, p_{n_E|\vec{i}}))$, where $\sum_j p_{j|\vec{i}} = 1$. 
Let $k_{j|\vec{i}}$ be the number of times outcome $j$ occurred for sequence $\vec{i}$. 
Then Algorithm~\ref{Alg:main} optimizes the objective function~\eqref{eq:objective function} on the estimated probabilities $y_{j|\vec{i}} = k_{j|\vec{i}}/\nsamples$. 

For the single qubit examples the batch size $\kappa = 50$ was chosen, while for the two qubit example we use $\kappa = 120$. The choice of batch size determines the number of values summed over in Eq.~\eqref{eq:objective function}. Therefore, the computation time of the objective function and its derivatives scales linearly in $\kappa$, making a small batch size favorable.
However, it cannot be set too small, otherwise the update directions become highly erratic, and no convergence is reached. 
A general rule of thumb is to set the batch size close to the number of free parameters in the model. 
Another hyperparameter is the damping value $\lambda$ for the saddle-free Newton method described in Algorithm~\ref{alg:SFN-update}. 
We find that a fixed value of $\lambda = 10^{-3}$ leads to the best results across the models tested. 

Judging whether \alg\ recovers a gate set by looking at the attained objective function value can only be done if the set of measured sequences is informationally complete. Then there is a unique (up to gauge) global minimum in the least squares minimization problem and the minimum corresponds to the true gate set in the limit of infinitely many samples per sequence. 

In Figure~\ref{fig:Obj-MVE-correlation} we take a look at the correlation between the final least square objective function value $\mathcal{L}(\hat{\mathcal{X}},\vec{y})$ and the \acl{MVE} $\ac{MVE} (\hat{\mathcal{X}},\mathcal{X})$. We see that for a low number of sequences ($10$-$20$), a low objective function value does not imply a low \ac{MVE}, yet for higher numbers of sequences, an objective function value below $10^{-3}$ implies an \ac{MVE} around $10^{-2}$. 
For sufficiently many sequences, the gray line indicating our success criterion clearly separates two clusters of points, meaning that no intermediate quality fits are found in our model space. In this sequence regime either the algorithm converges to a fit as good as the sample count allows, or it does not converge at all. 
Therefore, restarting the algorithm when an initialization turns out to be bad yields practically optimal results.
A thorough analysis of the probability of obtaining an informationally complete set of random sequences is left for future work. 

To give an intuition on how many initializations are required for \alg\ to converge, we can take a look at data from a modified XYI-model with gates $\{\mathds{1}, \e^{\i \frac{\alpha}{2} \sigma_y}, \e^{\i \frac{\alpha}{2} \sigma_x}\}$, simulating more difficult gate set conditioning. 
Figure~\ref{fig:Initialization histogram} shows histograms for the number of reinitializations needed for convergence. The data combines the results for $\alpha$ between $\pi/18$ and $\pi/2$, with depolarizing noise of strength $p = 0.001$ on each gate and $p = 0.01$ on the initial state, as well as a maximum of $100$ reinitializations. 

We observe that in $48.4\%$ of all cases convergence is reached on the first attempt, and in $90.4\%$ of cases $4$ or fewer
reinitializations are required. 
The histogram indicates that the chance of needing multiple reinitializations rapidly decays and that rank-1 optimization is more sensitive to bad initializations compared to rank-4 optimization. 

\begin{figure}[h!]
  \centering
    \includegraphics{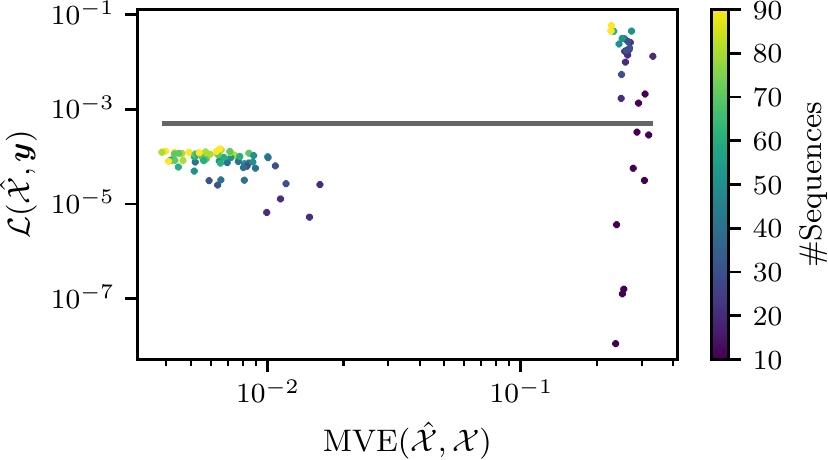}
    \caption{Least square objective function vs.\ \ac{MVE} at the end of the \alg-optimization, with the gray bar indicating the range of stopping values below which a run is considered successful.
    The plot illustrates that results with a large MVE also have a large objective function if enough sequences are measured.
    The experiment is the same as in the top left of Figure~\ref{fig:Phase transitions}, for $r_K = 4$. 
    The color of each point indicates the number of measured sequences on which \alg\ was run.
    }
  \label{fig:Obj-MVE-correlation}
\end{figure}

\begin{figure}[h!]
  \centering
    \includegraphics{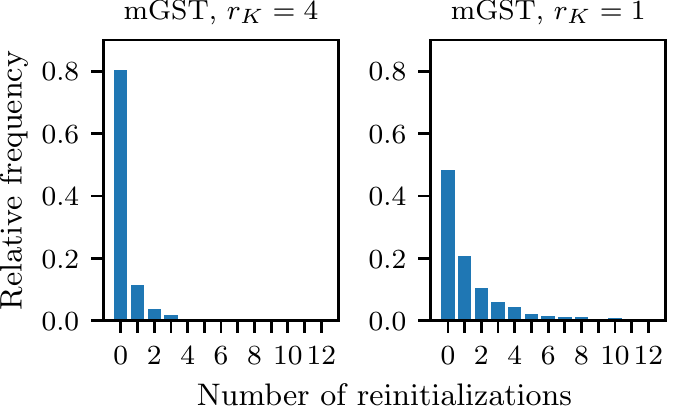}
  \caption[]{Relative frequencies of reinitialization counts for \alg run on modified XYI models (see main text for details).}
  \label{fig:Initialization histogram}
\end{figure}

\subsection{Runtime and scaling} \label{sec:runtime}

In order to assess the runtime scaling of Algorithm~\ref{Alg:main} in the problem parameters we identify the two most time-consuming steps as the computation of the second derivative terms in $H$ and the diagonalization of $H$ for the \ac{SFN} update on the gates. 

Recall that $r_K$ denotes the Kraus rank of the gate estimators, $n$ the number of gates in the gate set, $\ell$ the number of gates per sequence and $\kappa$ the number of sequences per batch. 
The computation of the second derivative terms scales as $\LandauO(\kappa \ell^3d^{6})$, while the eigendecomposition of $H$ scales as $\LandauO\bigl((2 n r_K d^2)^{3}\bigr)$. 
A computationally less expensive variant is to not optimize over all variables of the full gate tensor at once, but over the individual gates one after another. The complexity of the eigendecompositions then reduces to $\LandauO\bigl(n(2 r_K d^2)^{3}\bigr)$, which is beneficial for large gate sets. However, this approach also leads to slower convergence, and we choose to optimize over the full tensor by default.

Table \ref{table: standard deviation sources} contains runtimes for the system sizes studied in our numerical experiments. 
We find that reconstruction of single qubit gate sets can be achieved within seconds and low rank $2$-qubit gate sets can be reconstructed within minutes on a single modern 32-core CPU. The runtimes for the gradient descent method on the 2-qubit example were very fast up to high ranks, however we generally find that the default Newton method is more reliable across different gate set and sequence scenarios. If a good initialization is known, the runtime is reduced drastically. For instance the 3-qubit reconstruction done with prior knowledge about the gates in Section~\ref{sec: 3Q unitary} took 5 hours and 30 minutes to complete. 

\begin{table}[ht] 
\centering 
\begin{tabular}{l c |c |c |c } 
& \multicolumn{3}{c}{\textbf{\alg : Default (Newton)}} \\ 
\cmidrule(l){2-4} 
 & $r_K = 1$ & $r_K = 4$ & $r_K = 8$ & $r_K = 16$ \\ 
\midrule 
1Q & 4'' & 14'' & $\diagup$ & $\diagup$\\ 
2Q & 8'8''  & 39'50'' & 2h26'10'' & 9h1'43'' \\ 
3Q & 3d2h & $\diagup$ & $\diagup$ & $\diagup$\\ 
\end{tabular} \\
\begin{tabular}{l c |c |c |c } 
& \multicolumn{3}{c}{\textbf{\alg : Gradient descent}} \\ 
\cmidrule(l){2-4} 
 & $r_K = 1$ & $r_K = 4$ & $r_K = 8$ & $r_K = 16$ \\ 
\midrule 
1Q & 26'' & 17'' & $\diagup$ & $\diagup$\\ 
2Q & 5'40''  & 6'0'' & 5'51'' & 4'47'' \\ 
3Q & No conv. & $\diagup$ & $\diagup$ & $\diagup$\\ 
\end{tabular} \\
\begin{tabular}{l c } 
& \textbf{\pyGSTi}\\
\midrule 
1Q & 5'' \\ 
2Q & 1h6'48''  \\ 
3Q & No result after 4d14h \\ 
\end{tabular}
\caption{
Average runtimes for $1,2$ and $3$ qubits with selected Kraus ranks on a modern 32 core CPU. The average runtime was calculated as the average time until the first successful reconstruction with a random initialization was achieved.  The one qubit gate set is the same as in Figure~\ref{fig:shot noise}, the two qubit gate set the same as in Figure~\ref{fig:Phase transitions} and the three qubit gate set is the same as in Figure~\ref{fig:3Q_diamond} but without using any prior knowledge. For the gradient descent method in the three qubit scenario, no convergence was achieved within the maximum iteration limit. The average runtimes for \pyGSTi\ were determined on the same models and with the same sequences. 
} 
\label{table: standard deviation sources}
\end{table}

\section{Noise-mitgation for shadow estimation with compressive GST} 
\label{sec: shadows}
Gate set estimates provide a detailed picture of the imperfections that can inform prioritization and further experimental efforts \cite{BluGamNie17,dehollain2016optimization,2019SciA....5.5686S,2020NatCo..11..587Z,2020arXiv200701210C}. 
Coherent error estimates can often be directly corrected for by adjusting or optimizing the control.
We expect that the more economically accessible compressive estimates from mGST can be used in place of traditional GST estimates in the above applications, in particular when complemented with RB estimates of incoherent noise effects.
However, it is beyond the scope of this work to demonstrate mGST in a full engineering cycle of a quantum computing device.

\begin{figure*}[t]
  \centerline{\includegraphics{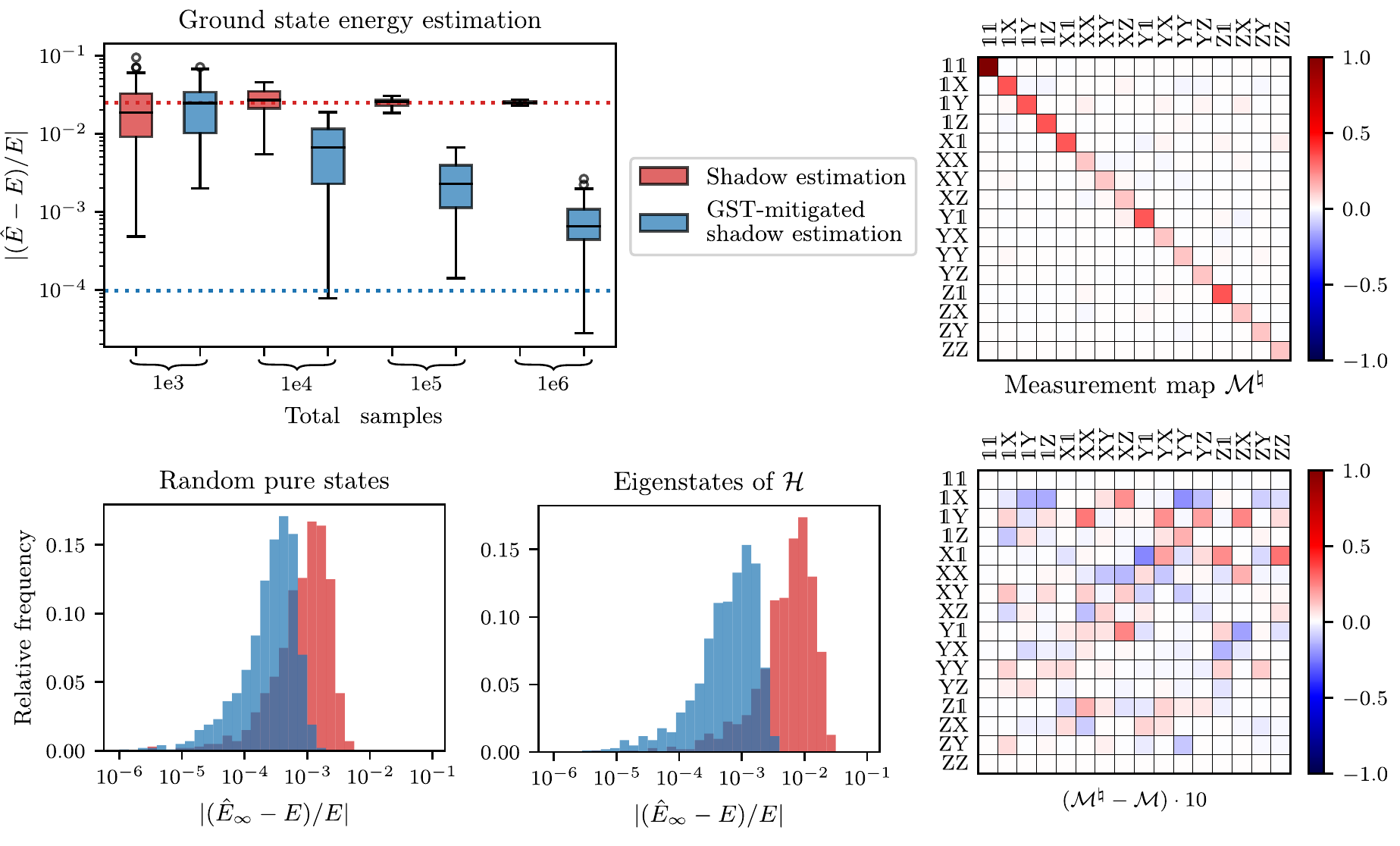}}
  \caption[]{
  Energy estimation for the 10 qubit Hamiltonian $\mathcal{H}=\frac{1}{2} \sum_{j=1}^{10}(\sigma_x^j \sigma^{j+1}_x+\sigma^j_y \sigma^{j+1}_y+\sigma^j_z \sigma^{j+1}_z - \sigma^j_z)$ with periodic boundary conditions. The \textbf{top left} shows the sample dependence of the relative accuracy $|(\hat E - E)/E|$ for estimating the ground state energy.
  The colored blocks extend from the 1st quartile to the 3rd quartile around the median (black line) of the data (50 repetitions per sample value). The whiskers extend from the 5th to the 95th percentile and the dashed lines indicate the infinite sample expectation values. On the \textbf{bottom left} two histograms are shown that compare the theoretical infinite sample energy estimates $\hat E_{\infty}$ (biases) for $1000$ random pure states and for all $1024$ eigenstates of the Hamiltonian, respectively. All simulations were done with noisy Clifford gates, whose average gate fidelity to their ideal counterparts is at $0.99 \pm10^{-3}$. On the \textbf{top right} the Pauli transfer matrix of a two-qubit effective measurement map $\mathcal{M}^\natural$ under this noise model is shown. The \textbf{bottom right} plot displays the difference between $\mathcal{M}^\natural$ and its noise-free counterpart $\mathcal{M}$. 
  Gate estimates for GST-mitigated shadow estimation were produced with \alg\ ($r_K = 2$), using $400$ random sequences equally distributed among sequences lengths $\{6,7,8,9\}$ with $10^4$ samples per sequence. The noise in the simulations is given by two-qubit random unitary noise $\e^{\i\gamma K}$ with $K \sim \mathrm{GUE}$. The error parameter is $\gamma = 0.14$ on $H$ and on $HS$, leading to the aforementioned average gate fidelities of $\sim 0.99$. For the bottom left histogram, random pure state were generated as $U\ket{0}$, with $U$ drawn according to the Haar measure. 
  }
  \label{fig:Shadows}
\end{figure*}

To still showcase the value of GST estimates, we sketch another novel application that can be demonstrated without simulating a whole engineering cycle.
Generally speaking, noise characterizations can be used to mitigate noise-induced biases in other quantum characterization protocols by adapting the classical post-processing \cite{2022arXiv221000921C,Endo21HybridQuantumClassical}.
Given a device that can repeatedly prepare a quantum state, a fundamental characterization task is to estimate the expectation values of observables from measurements.
In particular, an informationally complete measurement allows one to estimate arbitrary observables from the same data in the post-processing.
Such an informationally complete measurement can be implemented on a quantum computing device by measuring in sufficiently many random bases,
the prototypical example being measurements in random Pauli bases.
Ref.~\cite{Huang2020Predicting} showed how to derive optimal guarantees with exponential confidence for estimating multiple observables simultaneously using a median-of-means estimator and explicit bounds on the variance for random basis measurement that constitute unitary $3$-designs. 
They introduced the term `classical shadow' to refer to the elements of a dual frame of the informationally complete POVM corresponding to observed samples, see Appendix~\ref{app:Shadows} for a brief summary.
Importantly, one can often arrive at high precision estimates of observables long before one has measured all the informationally complete bases in multi qubit systems.

In practice however, implementing a random bases measurement, say, by applying a unitary rotation followed by a computational bases measurement will suffer from noise
from the gates and read-out.
This has motivated the development of robust variants of shadow estimation that either make use of simple depolarizing noise-models of known strength \cite{KohGrewal:2022:ClassicalShadows} or perform a separate RB-style experiment that estimates the depolarizing noise-strength induced by a gate-independent channel acting between the rotation and the measurement \cite{2021PRXQ....2c0348C}.
Using GST estimates provides a complimentary, flexible approach to mitigate even highly gate-dependent noise with finite correlations in shadow estimation.

We demonstrate how \ac{GST} estimates on $2$-qubit pairs can be used to calculate noise-robust classical shadow estimators in post-processing.
Our robust estimation scheme consists of two distinct stages each consisting of multiple steps:
(I) \emph{calibration stage}: (i) the local channels implementing each combination of two local gates are reconstructed with \alg; ii) the gauge of these gate estimates is matched to the gauge in which the ideal gates and the observables are given. For this step we use the gauge optimization provided by the \pyGSTi\ package \cite{Nielsen20pyGSTi}.
(iii) From the gauge-optimized channel estimates, we numerically calculate the \emph{effective measurement map} when implementing random Pauli measurements with the characterized noisy gates-set. 

(II) After calibration, the second stage is a \emph{shadow estimation} protocol consists of two separate phases:
(i) the \emph{data acquisition} by repeatedly measuring the unknown state of the quantum device in a randomly selected Pauli basis; 
(ii) the \emph{classical post-processing} where estimators of the observables are calculated using the data. We use the inverse of the effective measurement map from the calibration stage to calculate the empirical estimators.
We give a more detailed description of the individual steps of the procedure in Appendix \ref{app:Shadows}.

Using an empirically estimated effective measurement map instead of the ideal theoretical result is the essential modification compared to standard shadow estimation.
In this way, we also `invert' the effect of the noise on our estimator.
The right column of Figure~\ref{fig:Shadows} shows an effective measurement map implemented with imperfect gates.

As a proof of concept, we chose the following simple but practically relevant setup: 
Random local Pauli basis measurements are implemented by native measurements in the computational basis after rotating with a Hadamard gate $H$ (if the Pauli-X basis is to be measured) or a phase gate $S$ followed by a Hadamard gate (for measurement in the Pauli-Y basis). 
Since throughout the protocol the $S$-gate only turns up before application of the Hadamard gate, we treat the sequence $HS$ as a single gate. 

We assume that the dominant noise associated with the single qubit rotations of each local gate in the experiment stays confined to two neighboring qubits.
This assumption makes both the gate set estimation and the post-processing of the shadow estimation highly scalable.

Figure~\ref{fig:Shadows} shows the results of our scheme in simulations of the energy estimation of a Heisenberg Hamiltonian on a $10$-qubit system.
We observe that a using the estimated effective measurement map instead of the ideal theoretical one significantly reduces the relative error $|(\hat E - E)/E|$ between the estimated energy $\hat E$ and the true energy $E$ of a given state. There are two contributions to the relative error in a shadow estimation protocol: First, the statistical fluctuation from the randomness of both the Pauli basis selection and the single shot measurements. Second, the systematic bias introduced in the post-processing due to imperfect implementations of the measurements. 
The histograms in Figure~\ref{fig:Shadows} show the infinite measurement limit of the relative error and thus directly reflect the bias. We observe that for a fixed noise model, the magnitude of the bias depends heavily on the selected initial state, with relative errors being distributed over two orders of magnitude. When comparing the most likely errors between standard shadow estimation and GST-mitigated shadow estimation we find that using GST data leads to a reduction in relative error by half an order of magnitude for random pure states and an order of magnitude for eigenstates of the Hamiltonian. 
The simulation of the protocol in the top left of Figure~\ref{fig:Shadows} includes statistical fluctuations and showcases how estimates spread for different sample counts when the ground state energy is estimated. 
We find that from $10^4$ samples on, the GST-mitigated protocol yields significantly more accurate estimate. 

\section{Conclusion and outlook}
We have revisited the data processing task of \ac{GST} from a compressed sensing perspective regarding it as a highly structured and constrained tensor completion problem. 
In this formulation, we can naturally require the reconstructed gate set to be physical and, moreover, of low rank.
Compressive gate set tomography, thus, aims at extracting considerably fewer parameters of the gate set.  
At the same time we have argued that the low-rank 
approximation to the implementation of a gate set contains the most valuable information about experimental imperfections for the practitioner. 

The set of Kraus-operators of a low-rank gate can be regarded as isometries that make up the complex Stiefel manifold. 
This observation has motivated the solution of the compressive \ac{GST} data processing problem via geometrical optimization on the respective product manifolds. 
We have devised the optimization algorithm \alg\ that performs an adapted saddle-free Newton method on the manifold. 
To this end, we have derived the Riemannian Newton equation, Hessian equation and geodesic curves.

In numerical experiments we have studied the performance of the \alg\ algorithm. 
We have compared it to \pyGSTi, the state-of-the-art approach to the \ac{GST} data processing problem, in settings where both algorithms can be applied and using full rank \alg\ estimates. We have found that in these settings \alg\ matches the performance of \pyGSTi, while offering a larger model space, more flexibility in the sequence design and allowing for low rank assumptions. 
Moreover, we have demonstrated numerically
that making use of the low-rank constraints significantly reduces the required number of measured sequences and the run-time of the reconstruction algorithm 
for a standard single and two qubit model. 
Importantly, we have found that we can successfully reconstruct generic unitary channels and depolarizing noise of one- and two-qubit gate sets from \emph{random} gate sequences.  
This reduces
the demands of \ac{GST} both for 
experiments and 
classical post-processing: the data that compressive \ac{GST} requires is virtually identical with the experimental data produced by randomized benchmarking experiments.  The classical post-processing of \alg\ for a low-rank reconstruction of two qubit gate sets takes only minutes even on desktop hardware, compared to over an hour with \pyGSTi. 
We expect that this speedup and the low number of sequences required can lift 2-qubit \ac{GST} from being a protocol that is unpractical in many situations to one that is routinely applied, thus enabling it to be used in the engineering cycle for the design and calibration of gate sets. 

Furthermore, compressive \ac{GST} makes it feasible to perform self-consistent tomography on 3-qubit systems.  
Making use of often available prior knowledge about an initialization can further reduce the computing time.
We demonstrated this by performing tomography of unitary errors for 3-qubit gate sets, using only a small number of random gate sequences, and with the post-processing still running on desktop hardware in a few hours.  
We expect that even going slightly beyond three qubits is feasible by simply using 
more computing power. 
We leave it to future work to further tweak the numerical implementation in order to improve the scalability of the classical post-processing. We also expect that progressively longer sequences can be added at the end of our optimization method, much in the same fashion as in \pyGSTi, in order to further improve reconstruction accuracy of gate set estimates. 

To demonstrate the use of compressive \ac{GST} for error mitigation, we have introduced one novel application where low-rank \alg\ reconstructions are used to alleviate the effect of coherent errors in classical shadow estimation. 
The protocol uses a set of gates to implement basis changes before the measurement. 
We have demonstrated that with tomographic information on these gates through two-qubit compressive \ac{GST} yields more accurate ground state energy estimates in practically relevant regimes. This constitutes just one example where the full information of a low rank GST estimate is used to correct errors, and we expect that the reduced runtime requirements of \alg\ enable frequent use of \ac{GST} for error diagnosis and mitigation.

Finally, besides making \ac{GST} more applicable and flexible in practice, 
our reformulation is motivated by bringing it closer to theoretical recovery guarantees quantifying a required and sufficient number of random sequences for accurate reconstruction. 
Regarding the data processing of \ac{GST} as a translation-invariant matrix-product-state/tensor-train completion problem makes it more amenable to prove techniques from compressed sensing.  
For example, establishing local convergence guarantees for \alg\ would allow one to quantify the assumptions on the experimental implementation that justify certain initialization of the algorithm.  
We hope that our work can serve as a foundation and inspiration in the quest of establishing mathematically rigorous guarantees for \ac{GST}. 

\let\oldaddcontentsline\addcontentsline
\renewcommand{\addcontentsline}[3]{}
\hypertarget{Acknowledgments}{}
 \bookmark[level=section,dest=Acknowledgments]{Acknowledgments}
\section*{Acknowledgments}
\let\addcontentsline\oldaddcontentsline
We thank Lennart Bittel for helpful discussions and hints regarding optimization methods and their implementation and Markus Heinrich for a discussion on inverses of positive maps on operator spaces. We are particularly thankful to Kenneth Rudinger and Robin Blume-Kohout for valuable feedback on a previous version of the manuscript and help with finding an error in our code.
The work of RB and MK has been funded by 
the Deutsche Forschungsgemeinschaft (DFG, German Research Foundation) within the Emmy Noether program (grant number 441423094) and by 
the German Federal Ministry of Education and Research (BMBF) 
within the funding program ``quantum technologies -- from basic research to market'' via the joint project MIQRO (grant number 13N15522). IR acknowledges funding from the BMBF (DAQC) and the Einstein Foundation (Einstein Research Unit).

\newpage
\onecolumngrid
\section{Appendix}
\appendix
\renewcommand{\thesubsection}{\Alph{subsection}}
\renewcommand{\thethm}{\arabic{thm}}

In this appendix, we provide the mathematical details required for the saddle-free Newton method within the Riemannian optimization framework, see Appendices~\ref{app:geodesic}, \ref{app: complex riem newton}, and \ref{app:Euclidean gradient and Hessian}. 
Moreover, we compare the dependence of the \alg\ \ac{MVE} on the choice of objective function (mean squared error vs.\ maximum likelihood) in Appendix~\ref{app:MLE}. 

\subsection{Geodesics on the Stiefel manifold} \label{app:geodesic}
Edelman, Arias, and Smith \cite{edelman1998geometry} derived the geodesic on the real Stiefel manifold  by solving the respective geodesic equation. 
We now show that the simple generalization given in Eq.~\eqref{eq:geodesic def} is indeed the correct geodesic in the complex case. 
For a curve $K_t \equiv K(t)$ the general geodesic equation is \cite[Chapter~5.4, Proposition~5.3.2]{Absil09}
\begin{align} \label{eq: geodesic eq}
P_{T(K_t)}\left(\ddot{K}_t + C_{K_t}(\dot{K}_t,\dot{K}_t) \right) = 0 \,,
\end{align}
where the Christoffel symbol $C_{K_t}$ depends on the chosen metric. 
Here, we use the canonical metric
\begin{equation}
  \langle \Delta_1,\Delta_2\rangle_{K} = \Re\left\{\Tr(\Delta_1^\dagger \Gamma \Delta_2)\right\}\eqqcolon g(\Delta_1,\Delta_2)
\end{equation}
with $\Gamma = \mathds{1}-\frac{1}{2}K_tK_t^{\dagger}$. 
Using the Einstein summation convention, the Christoffel symbol at $K$ can be computed as 
\begin{align}
\left(C_{K_t}^k\right)_{ij} = \frac{1}{2} g^{-1}_{kl}\left(\frac{\partial g_{lj}}{\partial K_{ti}} + \frac{\partial g_{li}}{\partial K_{tj}} - \frac{\partial g_{ij}}{\partial K_{tl}} + \mathrm{c.c.} \right) \, ,
\end{align}
where $C^k_{K_t}$ is the $k$-th component of the Christoffel symbol at $K_t$ with respect to a basis $\{E_k, E_k^*\}_{k \in [Dd]}$ on the ambient space $\mathbb{C}^{D \times d}$.

\begin{lem}\label{lem:GeodesicEq}
The geodesic equation on the complex Stiefel manifold $\St(D,d)$ equipped with the canonical metric for the curve $K_t: \mathbb{R} \rightarrow \St(D,d)$ is given by 
\begin{equation} \label{eq: final geodesic eq}
P_{T(K_t)}\left( \ddot{K}_t + \dot{K}_t \dot{K}_t^{\dagger}K_t - K_t \dot{K}_t^{\dagger}\dot{K}_t - \dot{K}_t K_t^{\dagger} \dot{K}_t \right) = 0  \, .
\end{equation}
\end{lem}

\begin{proof}
By noting that $\Gamma^{-1} = \mathds{1}+K_tK_t^{\dagger}$ we can determine the function $g^{-1}(\Delta_1,\Delta_2)$ via the condition $g(g^{-1}(\Delta_1,\cdot \,),\Delta_2) = \Tr \left[\Delta_1^{\dagger} \Delta_2 + \Delta_2 \Delta_1^{\dagger}\right]$, meaning the inverse $g^{-1}$ would recover the standard symmetric inner product on $T_K \St(D,d)$. 
One can quickly verify that $g^{-1}(\Delta_1,\cdot \,) = \Gamma^{-1} \Delta_1$ satisfies this condition. 

We determine the derivatives of $g$ needed for the Christoffel symbol by explicitly writing out $g$ as 
\begin{align}
g(\Delta_1,\Delta_2) = \Tr \left[\Delta_1^{\dagger}\left(\mathds{1} - \frac{1}{2} K_tK_t^{\dagger}\right)\Delta_2
+ \Delta_2^{\dagger}\left(\mathds{1} - \frac{1}{2} K_tK_t^{\dagger}\right)\Delta_1\right] \, ,
\end{align}
from where we can find the derivatives by $K$ and $K^*$ as 
\begin{align}
\frac{\partial g_{ij}}{\partial K_{tl}} = \frac{\partial g}{\partial K_{tl}}(E_i,E_j) &= - \frac{1}{2} \Tr \left[E_i^{\dagger}E_l K^{\dagger}E_j
+ E_j^{\dagger}E_l K^{\dagger}E_i\right] \, , \\
\frac{\partial g}{\partial K_{tl}^*}(E_i,E_j) &= \left(\frac{\partial g}{\partial K_{tl}}(E_i,E_j)\right)^* \\
&= - \frac{1}{2} \Tr \left[E_i^{\dagger}K E_l^{\dagger}E_j
+ E_j^{\dagger}K E_l^{\dagger}E_i\right] \, .
\end{align}
With these derivatives we can calculate 
\begin{align}
C^k_{K_t}(\dot{K}_t,\dot{K}_t) &= \left(C_{K_t}^k\right)_{ij} (\dot{K}_t)_i (\dot{K}_t)_j \\
&= \frac{1}{2}g^{-1}_{kl}\left(\frac{\partial g_{lj}}{\partial K_{ti}} + \frac{\partial g_{li}}{\partial K_{tj}} - \frac{\partial g_{ij}}{\partial K_{tl}} + \mathrm{c.c.}\right) \dot{K}_{ti}\dot{K}_{tj} \\
&= \frac{1}{2}\left(\frac{\partial g}{\partial K_{ti}}\left(g^{-1}(E_k, \cdot \,),\dot{K}_t\right)\dot{K}_{ti}
+ \frac{\partial g}{\partial K_{tj}}\left(g^{-1}(E_k, \cdot \,),\dot{K}_t\right)\dot{K}_{tj}
- \frac{\partial g}{\partial K_{tl}}\left(\dot{K}_t,\dot{K}_t\right)(g^{-1}(E_k, \cdot \,))_l + \mathrm{c.c.}\right) \\
&= - \frac{1}{2} \Re \left\{\Tr\left[ 2E_k^{\dagger}\Gamma^{-1}\dot{K}_t K_t^{\dagger} \dot{K}_t + 2\dot{K}_t^{\dagger}\dot{K}_t K_t^{\dagger} \Gamma^{-1}E_k  - 2\dot{K}_t^{\dagger}\Gamma^{-1}E_kK_t^{\dagger}\dot{K}_t \right] \right\} \\
&= - \Re \left\{\Tr\left[ \dot{K}_t K_t \dot{K}^{\dagger}_t \Gamma^{-1}E_k + \dot{K}_t^{\dagger}\dot{K}_t K_t^{\dagger} \Gamma^{-1}E_k  - K_t^{\dagger}\dot{K}_t\dot{K}^{\dagger}\Gamma^{-1}E_k \right] \right\} \\
&= \left\langle \left(K_t^{\dagger}\dot{K}_t \dot{K}_t^{\dagger} - \dot{K}_t^{\dagger}\dot{K}_tK_t^{\dagger} - \dot{K}_t^{\dagger}K_t\dot{K}_t^{\dagger} \right)^{\dagger}, E_k\right\rangle \\
&= \left\langle \dot{K}_t \dot{K}_t^{\dagger}K_t - K_t \dot{K}_t^{\dagger}\dot{K}_t - \dot{K}_t K_t^{\dagger} \dot{K}_t,E_k\right\rangle \, ,
\end{align}
where we have used that $(\Gamma^{-1})^{\dagger} = \Gamma^{-1}$ and $\Re \Tr [X] = \Re \Tr [X^{\dagger}]$. We now first write out the geodesic equation \eqref{eq: geodesic eq} on the ambient space,
\begin{align}
\langle\ddot{K}_t,E_k\rangle  + C^k_{K_t}(\dot{K}_t,\dot{K}_t)
= \langle\ddot{K}_t,E_k\rangle + \left\langle \dot{K}_t \dot{K}_t^{\dagger}K_t - K_t \dot{K}_t^{\dagger}\dot{K}_t - \dot{K}_t K_t^{\dagger} \dot{K}_t,E_k\right\rangle = 0 \quad \forall E_k \, , \\
\end{align}
which is equivalent to 
\begin{align} \label{eq:geodesic eq without proj}
\ddot{K}_t + \dot{K}_t \dot{K}_t^{\dagger}K_t - K_t \dot{K}_t^{\dagger}\dot{K}_t - \dot{K}_t K_t^{\dagger} \dot{K}_t = 0  \, .
\end{align}
\end{proof}

To arrive at the geodesic equation \eqref{eq:geodesic def}, it remains to project the above equation onto the tangent space. 
Indeed, with the explicit form of the geodesic equation from Lemma~\ref{lem:GeodesicEq} we can show that the immediate generalization from the geodesic in the real case \cite{edelman1998geometry} gives a valid geodesic for the complex case. 

\begin{lem} \label{lem: geodesic lemma}
  The curve given by 
  \begin{equation} \label{eq: geodesic lemma}
  K_t =
  \begin{pmatrix} K & Q\end{pmatrix} \exp \left( t \begin{pmatrix}A & -R^{\dagger} \\ R & 0\end{pmatrix}\right) \begin{pmatrix}\mathds{1} \\ 0\end{pmatrix}\,
  \end{equation}
  is a geodesic on $\St(D,d)$, determined through the initial conditions $K_{t=0} = K$ and $\dot{K}_{t=0} = \Delta$, with $Q,R$ given by the QR decomposition of $(\mathds{1} - KK^{\dagger})\Delta$ and $A = K^{\dagger}\Delta$.
\end{lem}

\begin{proof} 
We recall that the ambient space splits into the tangent space and its orthogonal complement, the normal space. Therefore, the condition that the projection of the left-hand side onto the tangent space in Eq.~\eqref{eq: final geodesic eq} vanishes is equivalent to demanding that it lies solely in the normal space. 
If it is in the normal space, $K_t^{\dagger}$ applied from the left will yield a Hermitian matrix. 
We will now show that this is indeed the case. 
For that we first need to determine the first and second derivatives of $K_t$:
\begin{align} \label{eq:geodesic derivative}
\dot{K}_t &=
\begin{pmatrix} K & Q\end{pmatrix} \exp \left( t \begin{pmatrix}A & -R^{\dagger} \\ R & 0\end{pmatrix}\right) \begin{pmatrix}A & -R^{\dagger} \\ R & 0\end{pmatrix} \begin{pmatrix}\mathds{1} \\ 0\end{pmatrix} \\
&= \underbrace{K_t A}_{\dot{K}_1} + \underbrace{\begin{pmatrix} K & Q\end{pmatrix} \exp \left( t \begin{pmatrix}A & -R^{\dagger} \\ R & 0\end{pmatrix}\right) \begin{pmatrix} 0 \\ \mathds{1}\end{pmatrix}R}_{\dot{K}_2} \, , \\
\ddot{K}_t &=
\begin{pmatrix} K & Q\end{pmatrix} \exp \left( t \begin{pmatrix}A & -R^{\dagger} \\ R & 0\end{pmatrix}\right) \begin{pmatrix}A & -R^{\dagger} \\ R & 0\end{pmatrix}^2 \begin{pmatrix}\mathds{1} \\ 0\end{pmatrix} \\
&= \underbrace{K(t)(A^2 - R^{\dagger}R)}_{\ddot{K}_1} + \underbrace{\begin{pmatrix} K & Q\end{pmatrix} \exp \left( t \begin{pmatrix}A & -R^{\dagger} \\ R & 0\end{pmatrix}\right) \begin{pmatrix} 0 \\ \mathds{1}\end{pmatrix}RA}_{\ddot{K}_2} \, .
\end{align}
We immediately see that $K^{\dagger}_t \ddot{K}_1 = A^2 - R^{\dagger}R$, which is Hermitian, as $A$ is skew Hermitian. 
We will now show that $K^{\dagger}_t \ddot{K}_2 = 0$ starting with
\begin{align}
K^{\dagger}_t \ddot{K}_2 &= \begin{pmatrix} \mathds{1} & 0\end{pmatrix} \exp \left( - t \begin{pmatrix}A & -R^{\dagger} \\ R & 0\end{pmatrix}\right) \begin{pmatrix}K^{\dagger} \\ Q^{\dagger}\end{pmatrix} \begin{pmatrix} K & Q \end{pmatrix} \exp \left( t \begin{pmatrix}A & -R^{\dagger} \\ R & 0\end{pmatrix}\right) \begin{pmatrix}0 \\ \mathds{1} \end{pmatrix}RA \\
&=
\begin{pmatrix} \mathds{1} & 0\end{pmatrix} \exp \left( - t \begin{pmatrix}A & -R^{\dagger} \\ R & 0\end{pmatrix}\right) \begin{pmatrix}\mathds{1} & K^{\dagger}Q \\ Q^{\dagger}K & \mathds{1}\end{pmatrix} \exp \left( t \begin{pmatrix}A & -R^{\dagger} \\ R & 0\end{pmatrix}\right) \begin{pmatrix}0 \\ \mathds{1} \end{pmatrix}RA \\
&=
\begin{pmatrix} \mathds{1} & 0\end{pmatrix} \exp \left( - t \begin{pmatrix}A & -R^{\dagger} \\ R & 0\end{pmatrix}\right) \left(\begin{pmatrix}\mathds{1} & 0 \\ 0 & \mathds{1}\end{pmatrix} + \begin{pmatrix}0 & K^{\dagger}Q \\ Q^{\dagger}K & 0\end{pmatrix} \right) \exp \left( t \begin{pmatrix}A & -R^{\dagger} \\ R & 0\end{pmatrix}\right) \begin{pmatrix}0 \\ \mathds{1} \end{pmatrix}RA \\
&=
\begin{pmatrix} \mathds{1} & 0\end{pmatrix} \exp \left( - t \begin{pmatrix}A & -R^{\dagger} \\ R & 0\end{pmatrix}\right) \begin{pmatrix}0 & K^{\dagger}Q \\ Q^{\dagger}K & 0\end{pmatrix} \exp \left( t \begin{pmatrix}A & -R^{\dagger} \\ R & 0\end{pmatrix}\right) \begin{pmatrix}0 \\ \mathds{1} \end{pmatrix}RA \, . \label{eq:geodesic last line}
\end{align}
To simplify the last expression, we set 
\begin{align}
\begin{pmatrix}U_{00} & U_{01} \\ U_{10} & U_{11}\end{pmatrix} = \exp \left( t \begin{pmatrix}A & -R^{\dagger} \\ R & 0\end{pmatrix}\right)
\end{align}
and obtain $K^{\dagger}_t \ddot{K}_2 = (U_{00}^{\dagger}K^{\dagger}QU_{11} + U_{10}^{\dagger}Q^{\dagger}KU_{01})RA$. 
From the series representation of the matrix exponential we gather that $U_{11} = \mathds{1} + R \cdot X$ for some matrix $X$. Moreover $U_{10} = R \tilde{X}$ and $U_{10}^{\dagger} = \tilde{X}^{\dagger}R^{\dagger}$ for some $\tilde X$, leading to 
\begin{align} \label{eq:vanishing of geodesic term}
K^{\dagger}_t \ddot{K}_2 &= (U_{00}^{\dagger}K^{\dagger}Q + U_{00}^{\dagger}K^{\dagger}QRX + \tilde{X}^{\dagger}R^{\dagger}Q^{\dagger}KU_{01})RA = 0 \, ,
\end{align}
since $K^{\dagger}QR = K^{\dagger}(\mathds{1}-KK^{\dagger})\Delta = 0$. 

This shows that the $\ddot{K}_t$ lies in the normal space, leaving us with the terms in the geodesic equation~\eqref{eq: final geodesic eq} that depend only on $\dot{K}_t$:
\begin{align*}
K_t^{\dagger} (\dot{K}_t \dot{K}_t^{\dagger}K_t - K_t \dot{K}_t^{\dagger}\dot{K}_t - \dot{K}_t K_t^{\dagger} \dot{K}_t) &= 
K_t^{\dagger}\dot{K}_t \dot{K}_t^{\dagger}K_t - \dot{K}_t^{\dagger}\dot{K}_t - (K_t^{\dagger}\dot{K}_t)^2 \\
&= (A+K_t^{\dagger}\dot{K}_2)(A+K_t^{\dagger}\dot{K}_2)^{\dagger} - (A^{\dagger}A + A^{\dagger}K_t^{\dagger}\dot{K}_2 + \dot{K}_2^{\dagger}K_t A + \dot{K}_2^{\dagger}\dot{K}_2) \\
&\quad - (A+K_t^{\dagger}\dot{K}_2)^2 \\
&= AA^{\dagger} - A^{\dagger}A - \dot{K}_2^{\dagger}\dot{K}_2- A^2 \, . 
\end{align*}
The last line follows from $\dot{K}_2A = \ddot{K}_2$ and our previous observation that $K_t^{\dagger}\ddot{K}_2 = 0$, which implies that $K_t^{\dagger}\dot{K}_2 = 0$ as well. The remaining term $\dot{K}_2^{\dagger}\dot{K}_2$ can be computed similarly to $K_t^{\dagger}\ddot{K}_2$ and we obtain

\begin{align}
\dot{K}^{\dagger}_2 \dot{K}_2 &= 
R^{\dagger}\begin{pmatrix} 0 & \mathds{1}\end{pmatrix} \exp \left( - t \begin{pmatrix}A & -R^{\dagger} \\ R & 0\end{pmatrix}\right) \begin{pmatrix}\mathds{1} & K^{\dagger}Q \\ Q^{\dagger}K & \mathds{1}\end{pmatrix} \exp \left( t \begin{pmatrix}A & -R^{\dagger} \\ R & 0\end{pmatrix}\right) \begin{pmatrix}0 \\ \mathds{1} \end{pmatrix}R \\
&=
R^{\dagger}\begin{pmatrix} 0 & \mathds{1}\end{pmatrix} \exp \left( - t \begin{pmatrix}A & -R^{\dagger} \\ R & 0\end{pmatrix}\right) \left(\begin{pmatrix}\mathds{1} & 0 \\ 0 & \mathds{1}\end{pmatrix} + \begin{pmatrix}0 & K^{\dagger}Q \\ Q^{\dagger}K & 0\end{pmatrix} \right) \exp \left( t \begin{pmatrix}A & -R^{\dagger} \\ R & 0\end{pmatrix}\right) \begin{pmatrix}0 \\ \mathds{1} \end{pmatrix}R \\
&= R^{\dagger}R + 
R^{\dagger} \begin{pmatrix} 0 & \mathds{1}\end{pmatrix} \exp \left( - t \begin{pmatrix}A & -R^{\dagger} \\ R & 0\end{pmatrix}\right) \begin{pmatrix}0 & K^{\dagger}Q \\ Q^{\dagger}K & 0\end{pmatrix} \exp \left( t \begin{pmatrix}A & -R^{\dagger} \\ R & 0\end{pmatrix}\right) \begin{pmatrix}0 \\ \mathds{1} \end{pmatrix}R \\
&= R^{\dagger}R + R^{\dagger}(U_{11}^{\dagger} Q^{\dagger}K U_{01} + U_{01}^{\dagger}K^{\dagger} Q U_{11})R \\
&= R^{\dagger}R + R^{\dagger}\left((\mathds{1}+X^{\dagger}R^{\dagger}) Q^{\dagger}K U_{01} + U_{01}^{\dagger}K^{\dagger} Q (\mathds{1}+RX)\right)R \\
& = R^{\dagger}R \, ,
\end{align}
where we used again that $K^{\dagger}QR = R^{\dagger}Q^{\dagger}K = 0$ in the last line. 

We can now put all the terms obtained by multiplying Eq.~\eqref{eq:geodesic eq without proj} with $K_t^{\dagger}$ from the left together and find 
\begin{align}
K_t^{\dagger}\left(\ddot{K}_t + \dot{K}_t \dot{K}_t^{\dagger}K_t - K_t \dot{K}_t^{\dagger}\dot{K}_t - \dot{K}_t K_t^{\dagger} \dot{K}_t\right) &= A^2 - R^{\dagger}R - AA^{\dagger} - A^{\dagger}A - R^{\dagger}R - A^2 \\
&= -2R^{\dagger}R - A^{\dagger}A - AA^{\dagger} \, .
\end{align} 
We see that these remaining terms are Hermitian and therefore the left-hand side of Eq.~\eqref{eq:geodesic eq without proj} is in the normal space. 
\end{proof}

\subsection{Complex Newton equation} \label{app: complex riem newton}

In this section, we derive the Riemannian Hessian operator and solve the Hessian equation to obtain an update direction on the tangent space, which we can follow along the geodesic defined in Eq.~\eqref{eq:geodesic def}.
This can be done for each gate individually, or simultaneously over all gates, in which case we operate on the Cartesian product $\St(D,d)^{\times \ngates}$ of single Stiefel manifolds. 
We consider the latter case, whereby we obtain the single Stiefel Newton equation (Eq.~\eqref{eq:single gate final Newton}) as a byproduct. The method is based on the real case \cite{edelman1998geometry}. 
See also \cite{manton2002optimization} for a recent treatment of second order optimization on the complex Stiefel manifold, where instead of following geodesics, each optimization step is done in Euclidean space followed by a projection onto the manifold.  

First let us make a general observation that will be useful at several points. 
\begin{lem}[\cite{manton2002optimization}, Theorem 14] \label{lem:gradient formula}
Let $f: T_K\St(D,d) \rightarrow \mathbb{C}$ be a $\mathbb{C}$-linear function and let $\langle \argdot, \argdot \rangle_K$ be the canonical metric on $\St(D,d)$ as defined in Eq.~\eqref{eq:canonical metric def}. Then the solution to 
\begin{align}\label{eq:gradient formula}
\Re\left\{ f(\Delta)\right\} = \langle X,\Delta\rangle_K \quad \forall \Delta \in T_K \St(D,d)
\end{align}
is given by $X = F^* - KF^TK \in T_K \St(D,d)$, where $F$ is chosen such that $f(\Delta) = \Tr(F^T \Delta)$.
\end{lem}

\begin{proof}
It is straightforward to see that $X \in T_K \St(D,d)$ by applying the projector onto the tangent space:
$P_{T_K}(X) = F^* - KF^TK - \frac{1}{2}K(K^{\dagger}F^* + F^TK) + \frac{1}{2}K(F^TK + K^{\dagger}F^*) = \mathcal{F}^* - KF^TK\, .$ 

To show that $X$ solves Eq.~\eqref{lem:gradient formula}, we will use that $K^{\dagger}\Delta$ is skew Hermitian, as well as the fact that $\Re \Tr [HS] = 0$ for any skew Hermitian matrix $S$ and Hermitian matrix $H$.
Plugging $X = F^* - KF^TK $ into Eq.~\eqref{lem:gradient formula} we obtain
\begin{align*}
\langle X,\Delta\rangle_K &= \Re \Tr[X^{\dagger}(\mathds{1}-\frac{1}{2}KK^{\dagger})\Delta] \\
&= \Re \Tr[(F^T - K^{\dagger}F^*K^{\dagger})(\mathds{1}-\frac{1}{2}KK^{\dagger})\Delta] \\
&= \Re \Tr\left[\left(F^T - \frac{1}{2}(F^TK+K^{\dagger}F^*) K^{\dagger}\right)\Delta\right] \\
&= \Re \Tr [F^T \Delta] - \Re \Tr[\operatorname{herm}(F^TK)K^{\dagger}\Delta] \\
&= \Re \Tr [F^T \Delta] \, . 
\end{align*}
\end{proof}

Our goal is to simultaneously update all gates along the geodesic $\bigoplus_{i=1}^n \K_i(t) \in \St(D,d)^{\times \ngates}$, with the single Stiefel geodesics $\K_i(t)$ being given by Eq.~\eqref{eq:geodesic def}. 
We define the initial directions as $\Delta_i = \dot{\K}_i(0)$. The first step to identify the Riemannian gradient and Hessian is to compute the second order Taylor series expansion of $\mathcal{L}$ in $t$ at $t=0$. Using $(\K_i)_{lm} \equiv \K_{ilm}$ and Einstein notation we find
\begin{equation}
\begin{split}
\mathcal{L}(\K_1 &\oplus \dots \oplus \K_{\ngates};\K^*_1 \oplus \dots \oplus \K^*_n) = \mathcal{L}|_{t=0}
+ 2\left.\Re\left\{\frac{\partial \mathcal{L}}{\partial \K_{ilm}}\frac{\partial \K_{ilm}}{\partial t}\right\}\right|_{t=0}\cdot t 
\\&
+ 2\left.\Re\left\{\frac{\partial^2\Lo}{\partial \K_{jop} \partial \K_{ilm}}\frac{\partial \K_{jop}}{\partial t}\frac{\partial \K_{ilm}}{\partial t}
+ \frac{\partial^2\Lo}{\partial \K_{jop}^* \partial \K_{ilm}}\frac{\partial \K_{jop}^*}{\partial t}\frac{\partial \K_{ilm}}{\partial t}
+ \frac{\partial \Lo}{\partial \K_{ilm}}\frac{\partial^2 \K_{ilm}}{\partial t^2}
\right\}\right|_{t=0} \cdot t^2/2 \label{eq:taylor geodesic hess} 
+ \mathcal{O}(t^3) \, .
\end{split}
\end{equation}
We have $\left.\frac{\partial \K_{ilm}}{\partial t}\right|_{t=0} = (\Delta_i)_{lm}$ and define $(\mathcal{L}_{\K_i})_{lm} \coloneqq \frac{\partial \mathcal{L}}{\partial \K_{ilm}}$, so that we can write $\frac{\partial \mathcal{L}}{\partial \K_{ilm}}\frac{\partial \K_{ilm}}{\partial t} = \Tr(\mathcal{L}_{\K_i}^T\Delta_i)$ and $\frac{\partial \mathcal{L}}{\partial \K_{ilm}}\frac{\partial \K_{ilm}}{\partial t} \eqqcolon \mathcal{L}_{\K_i}[\Delta_i]$. In a similar fashion we define $\mathcal{L}_{\K_j\K_i}[\Delta_j,\Delta_i] \coloneqq \frac{\partial^2\Lo}{\partial \K_{jop} \partial \K_{ilm}}\frac{\partial \K_{jop}}{\partial t}\frac{\partial \K_{ilm}}{\partial t}$, where $\mathcal{L}_{\K_j\K_i}[\cdot,\cdot]$ is a bilinear function, which is symmetric per definition via the second derivative. For more details on how to compute these derivatives for the objective function used in the main text, see Appendix \ref{app:Euclidean gradient and Hessian}. \\

Before determining the relevant terms for the update on $\St(D,d)^{\times \ngates}$ we first consider the gradient and Hessian, as well as the Newton equation for a single variable $K \in \St(D,d)$, leaving all others constant. \\
The Riemannian gradient $G \in T_K \operatorname{St}(D,d)$ can be identified from the first order term in the Taylor expansion  via its definition \cite{manton2002optimization}
\begin{align} \label{eq:riem grad}
2*\Re\left\{\mathcal{L}_K\left[\Delta\right]\right\} = \langle G,\Delta \rangle_K \quad \forall \Delta \in T_K \St(D,d) \, .
\end{align}
The solution for $G$ in the canonical metric \eqref{eq:canonical metric def} is given by 
\begin{align}
G = 2\left(\Lo_K^* - K \Lo_K^T K\right) \, ,
\end{align}
as per Lemma \ref{lem:gradient formula}. 

\begin{lem}
  The Riemannian Hessian 
  $\operatorname{Hess}:T_K\St(D,d) \times T_K\St(D,d) \rightarrow \mathbb{R}$ of a function $\mathcal{L}: \St(D,d) \rightarrow \mathbb{R}$ with respect to the canonical metric on $\St(D,d)$ is given by
  \begin{equation} \label{eq:Hess lemma}
  \begin{aligned}
  \operatorname{Hess}(\Delta,\Omega) &= 2 \Re\left\{\Lo_{KK}[\Delta,\Omega] + \Lo_{K^*K}[\Delta^*,\Omega]\right\} \\
  &+ \Re\left\{\Tr\left[\Lo_K^T(\Delta K^{\dagger}\Omega + \Omega K^{\dagger}\Delta)\right]
  -\Tr\left[\Lo_K^T K(\Delta^{\dagger}\Pi\Omega +\Omega^{\dagger}\Pi\Delta) \right]\right\} \,.
  \end{aligned}
  \end{equation}
\end{lem} 

\begin{proof} 
According to \cite[Proposition~5.5.5]{Absil09}, we can compute $\mathrm{Hess}(\Delta,\Omega)$ via 
\begin{align} \label{eq:Hessian def}
\operatorname{Hess}(\Delta, \Omega) &= \frac{1}{2} \frac{\rmd^2}{\rmd t^2} \left.\left[ \Lo(K(t(\Delta + \Omega))) - \Lo(K(t\Delta)) - \Lo(K(t\Omega))\right]\right|_{t = 0} 
\end{align}
where $K(t \Delta)$ satisfies $\dot{K}(t\Delta)|_{t=0} = \Delta$ (see \cite{edelman1998geometry} for a discussion of the real case). The individual terms in Eq.\ \eqref{eq:Hessian def} can be determined from our general Taylor approximation in Eq.\ \eqref{eq:taylor geodesic hess}, i.e.\ with $i=j=1$, if we take $K = K_1$. The term $\left.\frac{\partial \Lo}{\partial K}\left[\frac{\partial^2 K(t\Delta)}{\partial t^2}\right]\right|_{t=0}$ contains second derivatives of the geodesic given in Lemma \ref{lem: geodesic lemma}, which we write out next. $\ddot{K}(t)$ is given by
\begin{align*}
\ddot{K}(t) &=
\begin{pmatrix} K & Q\end{pmatrix} \exp \left( t \begin{pmatrix}A & -R^{\dagger} \\ R & 0\end{pmatrix}\right) \begin{pmatrix}A & -R^{\dagger} \\ R & 0\end{pmatrix}^2 \begin{pmatrix}\mathds{1} \\ 0\end{pmatrix} \\
&= \begin{pmatrix} K & Q\end{pmatrix} \exp \left(t \begin{pmatrix}A & -R^{\dagger} \\ R & 0\end{pmatrix}\right)
\begin{pmatrix}A^2 - R^{\dagger}R \\ RA\end{pmatrix} \, .
\end{align*}
It follows using $QR = (\mathds{1} - KK^{\dagger}) \Delta \eqqcolon \Pi \Delta$ and $A = K^{\dagger}\Delta$ from the definition of the geodesic, that 
\begin{align*}
\ddot{K}(0) &= K(A^2 - R^{\dagger}R) + QRA \\
&= K(A^2 - R^{\dagger}Q^{\dagger}QR) + \Pi \Delta K^{\dagger} \Delta \\
&= K(K^{\dagger}\Delta K^{\dagger}\Delta - \Delta^{\dagger}\Pi^{\dagger}\Pi\Delta) + \Pi\Delta K^{\dagger}\Delta \\
&= K(K^{\dagger}\Delta K^{\dagger}\Delta - \Delta^{\dagger}\Pi\Delta) + \Delta K^{\dagger}\Delta - KK^{\dagger}\Delta K^{\dagger}\Delta \\
&= \Delta K^{\dagger}\Delta - K\Delta^{\dagger}\Pi \Delta \, .
\end{align*}{}
Putting the terms together, we arrive at
\begin{equation} \label{eq:second derivatives}
\frac{\rmd^2}{\rmd t^2} \left.\Lo(K(t\Delta))\right|_{t = 0}  = 2 \Re\left\{\Lo_{KK}[\Delta,\Delta] + \Lo_{K^*K}[\Delta^*,\Delta]\right\} + \Re\left\{\Tr\left[\Lo_K^T(\Delta K^{\dagger}\Delta - K \Delta^{\dagger} \Pi\Delta)\right]\right\} \, .
\end{equation}
The terms involving $\Lo_{KK}$ and $\Lo_{K^*K}$ satisfy $\Lo_{KK}[\Delta,\Omega] = \Lo_{KK}[\Omega,\Delta]$ and $\Lo_{K^*K}[\Delta^*,\Omega] = \Lo_{K^*K}[\Omega^*,\Delta]$, by the symmetry of second derivatives. 
Using this symmetry property we obtain the full Hessian \eqref{eq:Hessian def}, which turns out to be
\begin{equation} \label{eq:Hess bilinear}
\begin{aligned}
\operatorname{Hess}(\Delta,\Omega) &= 2 \Re\left\{\Lo_{KK}[\Delta,\Omega] + \Lo_{K^*K}[\Delta^*,\Omega]\right\} \\
&+ \Re\left\{\Tr\left[\Lo_K^T(\Delta K^{\dagger}\Omega + \Omega K^{\dagger}\Delta)\right]
-\Tr\left[\Lo_K^T K(\Delta^{\dagger}\Pi\Omega +\Omega^{\dagger}\Pi\Delta) \right]\right\} \, ,
\end{aligned} 
\end{equation} 
where it is helpful to note that Eq.~\eqref{eq:Hess bilinear} is related to Eq.~\eqref{eq:second derivatives} via a symmetrization of the $\mathcal{L}_K$ term. 
\end{proof}

\begin{thm} \label{thm: Newton step}
  Let $\vvec(\Delta)$ be the row major vectorization of $\Delta \in T_K \St(D,d)$. 
  Furthermore, let $T$ and $\tilde{\Lo}_{KK}$ be defined by $T \vvec(X) = \vvec(X^T)$ and $\Lo_{KK}(\Delta,\,\cdot \,\,) = \tilde{\Lo}_{KK}^T \vvec(\Delta)$.
  Then the solution $\Delta$ of the linear equation in $\vvec(\Delta)$ and $\vvec(\Delta^*)$ given by
  \begin{align}\label{eq:single gate final Newton}
  \left( \tilde{\Lo}_{K^*K}^{\dagger} - (K\otimes K^T)T\tilde{\Lo}_{KK}^T - \frac{1}{2}\mathds{1}\otimes (K^T\Lo_K) - \frac{1}{2}(K\Lo_K^T)\otimes \mathds{1} - \frac{1}{2}\Pi\otimes (\Lo_K^{\dagger}K^*) \right)\vvec(\Delta)\\
  + 
  \left( \tilde{\Lo}_{KK}^{\dagger} - (K\otimes K^T)T\tilde{\Lo}_{K^*K}^T + \frac{1}{2}(\Lo_K^* \otimes K^T)T + \frac{1}{2}(K\otimes\Lo_K^{\dagger})T \right)\vvec(\Delta^*) = -\frac{1}{2} \vvec(G) \, 
  \end{align}
  is the update direction along the geodesic given in Lemma \ref{lem: geodesic lemma} for the complex Newton method of a real function $\mathcal{L}$ at position $K \in \St(D,d)$.
\end{thm}

\begin{proof} 
The update direction $\Delta$ for the standard Newton method \cite{edelman1998geometry} is determined through the equation
\begin{align}
\operatorname{Hess}(\Delta,\Omega) = - \langle G,\Omega\rangle_K \quad \forall \Omega \in T_K \St(D,d) \, ,
\end{align}
which can be solved by rewriting the left-hand side as $\operatorname{Hess}(\Delta,\Omega) = \langle f(\Delta),\Omega\rangle_K$ (for some yet to be determined $f$) and setting $\Omega = P_T(X)$ with arbitrary matrix $X$. 
This leads us to 
\begin{align}
\langle f(\Delta),P_T(X)\rangle_K &= - \langle G,P_T(X)\rangle_K \\
\langle P_T(f(\Delta)),X\rangle_K &= - \langle G,X\rangle_K \quad \forall X \in \mathbb{C}^{D\times d} \, ; \label{eq: Newton2}
\end{align}
in the second line the scalar product is extended from the canonical scalar product initially defined on $T_K\St(D,d)$ to $\mathbb{C}^{D\times d}$ and we will use the same notation for both. 
The second line follows from the fact that any matrix $X$ can be decomposed as $X = P_T(X) + P_N(X)$ and from $\langle A,B\rangle_K = 0$ for $A \in T_K \St(D,d)$ and $B\in N_K \St(D,d)$. To determine $f(\Delta)$ we split it into three terms $f(\Delta) = f_{KK}(\Delta) + f_{K^*K}(\Delta) + f_{K}(\Delta)$, where $f_{KK}(\Delta), f_{K^*K}(\Delta)$ and $f_{K}(\Delta)$ depend only on $\mathcal{L}_{KK}, \mathcal{L}_{K^*K}$ and $\mathcal{L}_{K}$ respectively (compare Eq.~\eqref{eq:Hess bilinear}).\\ 
We first look at the term $2 \Re\left\{\Lo_{KK}[\Delta,\Omega]\right\} = 2 \Re\left\{\Tr\left(\Lo_{KK}[\Delta, \,\cdot\,\,]^T \Omega\right)\right\}$, where $\Lo_{KK}[\Delta,\,\cdot\,\,]$ is in $\mathbb{C}^{D \times d}$.

To solve $2 \Re\left\{\Tr\left(\Lo_{KK}[\Delta, \,\cdot\,\,]^T \Omega\right)\right\} = \langle f_{KK}(\Delta),\Omega\rangle_K$ for all $\Omega \in T_K \St(D,d)$ and to find $f_{KK}$ we use Lemma \ref{lem:gradient formula} and obtain
\begin{align}\label{eq:fkk}
f_{KK}(\Delta) = 2\left(\Lo_{KK}[\Delta,\,\cdot\,\,]^* - K\Lo_{KK}[\Delta,\,\cdot\,\,]^TK  \right)\, .
\end{align}
The same argument can be made for the $\Lo_{K^*K}$ term, leading to
\begin{align}\label{eq:fk*k}
f_{K^*K}(\Delta^*) = 2\left(\Lo_{K^*K}[\Delta^*,\,\cdot\,\,]^* - K\Lo_{K^*K}[\Delta^*,\,\cdot\,\,]^TK  \right)\, .
\end{align}
To identify $f_K(\Delta)$ we rewrite the second line in \eqref{eq:Hess bilinear} as follows:
\begin{align*}
\Re&\left\{\Tr\left[\Lo_K^T(\Delta K^{\dagger}\Omega +\Omega K^{\dagger}\Delta)\right]
-\Tr\left[\Lo_K^T K(\Delta^{\dagger}\Pi\Omega + \Omega^{\dagger}\Pi\Delta) \right]\right\} \\
&= \Re \left\{\Tr\left[\left(\Lo_K^T\Delta K^{\dagger} + K^{\dagger} \Delta\Lo_K^T -\Lo_K^T K \Delta^{\dagger}\Pi -  (\Pi\Delta\Lo_K^T K)^{\dagger}\right)\Omega\right]\right\} \\
&\overset{!}{=} \Re\left\{\Tr \left[f_K(\Delta)^{\dagger} \Gamma \Omega \right]\right\} \, ,
\end{align*}
where we used $\Re\left\{\Tr\left[AB^{\dagger}\right]\right\} = \Re\left\{\Tr\left[A^{\dagger}B\right]\right\}$.
Thus we find 
\begin{align*}
f_K(\Delta) &= \left[\left(\Lo_K^T\Delta K^{\dagger} + K^{\dagger} \Delta\Lo_K^T -\Lo_K^T K \Delta^{\dagger}\Pi -  (\Pi\Delta\Lo_K^T K)^{\dagger}\right)\Gamma^{-1}\right]^{\dagger} \\
&= 2K\Delta^{\dagger}\Lo_K^* + \Gamma^{-1}\Lo_K^*\Delta^{\dagger}K - \Pi\Delta K^{\dagger}\Lo_K^* - \Pi \Delta\Lo_K^T K \, 
\end{align*}
by using $\Pi = \Pi^{\dagger}, \,\Pi \Gamma^{-1} = \Gamma^{-1}\Pi = \Pi, \, K^{\dagger}\Gamma^{-1} = 2K^{\dagger}$. \\
Eq.~\eqref{eq: Newton2} implies $P_T\left(f(\Delta\right)) = - G$, and it remains to compute $P_T(f_K(\Delta))$, as $P_T(f_{KK}(\Delta)) = f_{KK}(\Delta)$ and $P_T(f_{K^*K}(\Delta^*)) = f_{K^*K}(\Delta^*)$ (see Lemma \ref{lem:gradient formula}). After a straightforward computation using $\Gamma^{-1} = \mathds{1} + KK^{\dagger}, \Pi \Gamma^{-1} = \Pi, P_{T}(\Pi Z) = P_{T}(Z)$, as well as $K^{\dagger}\Pi = \Pi K = 0$, we get
\begin{align}\label{eq:fk}
P_T(f_K(\Delta)) = - \Pi \Delta K^{\dagger} \Lo_K^* - 2 \operatorname{skew}(\Delta \Lo_K^T)K - 2K\operatorname{skew}(\Lo_K^T\Delta) \, ,
\end{align} 
with $\operatorname{skew}(A) = (A - A^{\dagger})/2$. \\
Finally, by plugging Eqs.~\eqref{eq:fkk}, \eqref{eq:fk*k} and \eqref{eq:fk} into Eq.~\eqref{eq: Newton2}, we obtain the Newton equation
\begin{align}
\Lo_{KK}[\Delta,\,\cdot\,\,]^* - K\Lo_{KK}[\Delta,\,\cdot\,\,]^TK + \Lo_{K^*K}[\Delta^*,\,\cdot\,\,]^* - K\Lo_{K^*K}[\Delta^*,\,\cdot\,\,]^TK  \\ - \frac{1}{2} \Pi \Delta K^{\dagger} \Lo_K^* - \operatorname{skew}(\Delta \Lo_K^T)K - K\operatorname{skew}(\Lo_K^T\Delta) = -G/2 \, .
\end{align}
This is a linear equation in $\Delta$ and $\Delta^*$ that can be solved via rewriting it as an equation in $\operatorname{vec}(\Delta)$. Using row-major vectorization with $\vvec(AXB) = (A\otimes B^T) \vvec(X)$, the matrices $T$ and $\tilde{\Lo}_{KK}$ defined by $T \vvec(X) = \vvec(X^T)$ and $\Lo_{KK}(\Delta,\,\cdot \,\,) = \tilde{\Lo}_{KK}^T \vvec(\Delta)$, we arrive at the final equation for the single gate case

\begin{align*}
\left( \tilde{\Lo}_{K^*K}^{\dagger} - (K\otimes K^T)T\tilde{\Lo}_{KK}^T - \frac{1}{2}\mathds{1}\otimes (K^T\Lo_K) - \frac{1}{2}(K\Lo_K^T)\otimes \mathds{1} - \frac{1}{2}\Pi\otimes (\Lo_K^{\dagger}K^*) \right)\vvec(\Delta)\\
+ 
\left( \tilde{\Lo}_{KK}^{\dagger} - (K\otimes K^T)T\tilde{\Lo}_{K^*K}^T + \frac{1}{2}(\Lo_K^* \otimes K^T)T + \frac{1}{2}(K\otimes\Lo_K^{\dagger})T \right)\vvec(\Delta^*) = -\frac{1}{2} \vvec(G) \, .  
\end{align*}  
\end{proof}

Now for the simultaneous optimization over all gates on $\operatorname{St}(D,d)^{\times \ngates}$, the Hessian as defined in Eq.\ \eqref{eq:Hessian def} is determined by including all terms in Eq.\ \eqref{eq:taylor geodesic hess}. 
The Newton equation reads
\begin{align}\label{eq:full Newton def} 
\operatorname{Hess}(\Delta_1 \oplus \dots \oplus \Delta_{\ngates};\Omega_1 \oplus \dots \oplus \Omega_n)
= -\sum_{i = 1}^{\ngates} \langle G_i,\Omega_i\rangle_{\K_i}
\end{align}
for all $\Omega_i \in T_{\K_i}\St(D,d)$. The terms in Eq.\ \eqref{eq:taylor geodesic hess} where $i=j$ are obtained from the single variable case. The mixed variable terms $f_{\K_i\K_j}(\Delta_i)$ and $f_{\K_i^*\K_j}(\Delta_i^*)$ still need to be determined.
Analogous to Eq.~\eqref{eq:fk*k} we need to solve 
\begin{align*} \label{eq:mixed terms}
2 \Re\left\{\Tr\left(\Lo_{\K_i^*\K_j}[\Delta_i^*, \,\cdot\,\,]^T \Omega_j\right)\right\} &= \langle f_{\K_i^*\K_j}(\Delta_i),\Omega_j\rangle_{\K_j} \quad \forall \Omega_j \in T_{\K_j}\St \quad \mathrm{and} \\
2 \Re\left\{\Tr\left(\Lo_{\K_i\K_j}[\Delta_i, \,\cdot\,\,]^T \Omega_j\right)\right\} &= \langle f_{\K_i\K_j}(\Delta_i),\Omega_j\rangle_{\K_j}  \quad \forall \Omega_j \in T_{\K_j}\St \, .
\end{align*}
We can use Lemma \ref{lem:gradient formula} again and obtain
\begin{align*}
 f_{\K_i^*\K_j}(\Delta_i) &= 2\left(\Lo_{\K_i^*\K_j}^*[\Delta_i,\,\cdot\,\,] - \K_j\Lo_{\K_i^*\K_j}[\Delta_i^*,\,\cdot\,\,]^T\K_j  \right)\, \quad \mathrm{and} \\
 f_{\K_i\K_j}(\Delta_i) &= 2\left(\Lo_{\K_i\K_j}^*[\Delta_i^*,\,\cdot\,\,] - \K_j\Lo_{\K_i\K_j}[\Delta_i,\,\cdot\,\,]^T\K_j  \right)\, ,
\end{align*}
which satisfy $f_{\K_i^*\K_j}(\Delta_i^*) \in T_{\K_j}\St(D,d)$ and $f_{\K_i\K_j}(\Delta_i) \in T_{\K_j}\St(D,d)$. 
The full Newton equation on $\operatorname{St}(D,d)^{\times \ngates}$ in vectorized form then reads
\begin{equation}\label{eq:final full Newton}
\begin{aligned}
\bigoplus_{i= 1}^{\ngates} \left( \tilde{\Lo}_{\K_i^*\K_i}^{\dagger} - (\K_i\otimes \K_i^T)T\tilde{\Lo}_{\K_i\K_i}^T - \frac{1}{2}\mathds{1}\otimes (\K_i^T\Lo_{\K_i}) - \frac{1}{2}(\K_i\Lo_{\K_i}^T)\otimes \mathds{1} - \frac{1}{2}\Pi\otimes (\Lo_{\K_i}^{\dagger}\K_i^*) \right)\vvec(\Delta_i)\\
+ 
\bigoplus_{i=1}^{\ngates} \sum_{j: j\neq i} \left( \tilde{\Lo}^{\dagger} _{\K_j^*\K_i} - (\K_i \otimes \K_i^T)T\tilde{\Lo}^T_{\K_j\K_i} \right) \vvec(\Delta_j) \\
+
\bigoplus_{i= 1}^{\ngates} \left( \tilde{\Lo}_{\K_i\K_i}^{\dagger} - (\K_i\otimes \K_i^T)T\tilde{\Lo}_{\K_i^*\K_i}^T + \frac{1}{2}(\Lo_{\K_i}^* \otimes \K_i^T)T + \frac{1}{2}(\K_i\otimes\Lo_{\K_i}^{\dagger})T \right)\vvec(\Delta_i^*) \\
+
\bigoplus_{i=1}^{\ngates} \sum_{j: j\neq i} \left( \tilde{\Lo}^{\dagger} _{\K_j\K_i} - (\K_i \otimes \K_i^T)T\tilde{\Lo}^T_{\K_j^*\K_i} \right) \vvec(\Delta_j^*) \\
= -\frac{1}{2} \bigoplus_{i = 1}^{\ngates} \vvec(G_i)\, .
\end{aligned}
\end{equation}
We are now faced with an equation of the type $A\vec x + B \vec{x^*} = \vec c$, a solution to which can be obtained by solving
\begin{align}
\begin{pmatrix} A &B \\ B^* & A^* \end{pmatrix} \begin{pmatrix} \vec x \\ \vec{x^*}\end{pmatrix} = \begin{pmatrix} \vec c \\ \vec{c^*}\end{pmatrix} \, .
\end{align}
Eq.~\eqref{eq:final full Newton} is an equation on the tangent space and can be solved by finding a basis therein. However, in order to avoid a basis change at every step, we choose to solve it on the ambient space by setting $\Delta_i = P_{T_i}(\Delta_i)$ and $\Delta_i^* = P_{T_i}^*(\Delta_i^*)$. The matrix equation for the update directions $\Delta_i$ is now given by
\begin{align} \label{eq:Newton matrix form}
\renewcommand{\arraystretch}{3}
\underbrace{
\begin{pmatrix}
H_{G \leftarrow \Delta} \bigoplus_i P_{T_i} & H_{G \leftarrow \Delta^*}\bigoplus_i P^*_{T_i} \\
H^*_{G \leftarrow \Delta^*}\bigoplus_i P^*_{T_i} & H^*_{G \leftarrow \Delta}\bigoplus_i P_{T_i} \\
\end{pmatrix}
}_{ \eqqcolon H}
\renewcommand{\arraystretch}{1}
\begin{pmatrix}
\vvec(\Delta_1) \\ \dots \\ \vvec(\Delta_{\ngates}) \\ \vvec(\Delta_1^*) \\ \dots \\ \vvec(\Delta_{\ngates}^*)
\end{pmatrix}
 = - \frac{1}{2}
 \begin{pmatrix}
\vvec(G_1) \\ \dots \\ \vvec(G_{\ngates}) \\ \vvec(G_1^*) \\ \dots \\ \vvec(G_{\ngates}^*)
\end{pmatrix} \, ,
\end{align}
where the sub matrices $H_{G \leftarrow \Delta}$ and $H_{G \leftarrow \Delta^*}$ can be identified from Eq.~\eqref{eq:final full Newton}. \\
The update directions for the saddle-free Newton method \ref{alg:SFN-update} are calculated by applying 
\begin{align}
\left(
\left|
H
\right|
+ \lambda \mathds{1}
\right)^{-1}
\end{align}
to the right-hand side of Eq.~\eqref{eq:Newton matrix form}. In practice, we use $\frac{1}{2}(H + H^{\dagger})$ as the Hessian, since there exist more efficient methods for diagonalizing a Hermitian matrix compared to an arbitrary matrix.

\subsection{Complex Euclidean gradient and Hessian} \label{app:Euclidean gradient and Hessian}

\begin{figure}[ht]
  \includegraphics{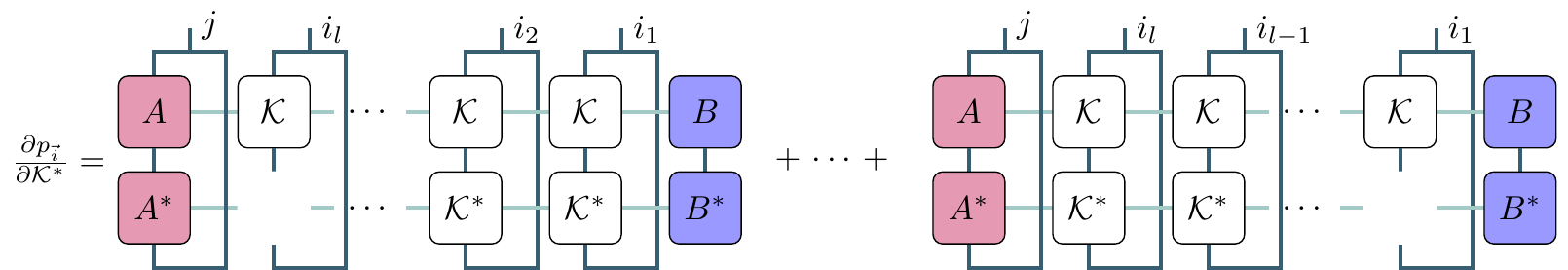}
  \caption{Tensor network representation of first derivative.}
  \label{fig:derivative}
\end{figure}

To compute the Riemannian Hessians for optimization on the Stiefel manifolds the complex Euclidean gradients and Euclidean Hessians are needed. We shall here go into more detail of their derivations for the least squares objective function.
This section also goes into detail as to how the terms $\mathcal{L}_{\K\K}$, $\mathcal{L}_{\K^*\K}$ and their conjugates from Appendix \ref{app: complex riem newton} are calculated. 

As the tensor network whose contraction yields $p_{j|\vec i}$ in the cost function~\eqref{eq:objective function} is parameterized in terms of matrix variables and their conjugates, we use Wirtinger calculus for the derivatives and treat the conjugate variables as independent. A method for finding all the relevant terms in the Hessian for a scalar function of complex matrix variables is outlined, e.g.\ in Ref.~\cite{hjorungnes2011complex}, and we will summarize it in the following. 
See also Ref.~\cite{van1994complex} for a short derivation of the complex derivative and Hessian in the context of optimization in complex Euclidean space. 

Our objective function is in general not analytic, as is the case with real valued functions of complex variables. This can be seen for the simplest case with one unitary gate $U$ and $\rho = E = \ketbra{0}{0}$, where $\mathcal{L} = \left|\sandwich{0}{U}{0}\right|^4 = |U_{00}|^4 = (U_{00}*U_{00}^*)^2$. 
However the derivatives w.r.t.\ the real and imaginary parts of the matrix variables exist and one can define formal derivatives for $f:\mathbb{C}^{M\times N}\times\mathbb{C}^{M\times N} \rightarrow \mathbb{R}$ via
\begin{align*}
\frac{\partial f(Z,Z^*)}{\partial Z} &\coloneqq \frac{\partial f(Z,Z^*)}{\partial \mathrm{Re}[Z]} - \i \, \frac{\partial f(Z,Z^*)}{\partial \mathrm{Im}[Z]} \, , \\
\frac{\partial f(Z,Z^*)}{\partial Z^*} &\coloneqq \frac{\partial f(Z,Z^*)}{\partial \mathrm{Re}[Z]} + \i \, \frac{\partial f(Z,Z^*)}{\partial \mathrm{Im}[Z]} \, , 
\end{align*}
where $\frac{\partial f(Z,Z^*)}{\partial Z} \in \mathbb{C}^{M\times N}$ with $\left(\frac{\partial f(Z,Z^*)}{\partial Z}\right)_{ij} = \frac{\partial f(Z,Z^*)}{\partial Z_{ij}}$. These formal derivatives have nice properties, for instance $\frac{\partial f(Z,Z^*)}{\partial Z^*}$ is the direction of maximum increase of $f$ and $\frac{\partial f(Z,Z^*)}{\partial Z^*} = 0$ identifies a stationary point of $f$, see e.g.\ Ref.~\cite[Theorems~3.2 and~3.4]{hjorungnes2011complex}. 
Furthermore, the product rule and the chain rule apply as they do for real valued matrix variables.

As laid out in Ref.~\cite[Lemma~5.2]{hjorungnes2011complex}, we can write the second order Taylor series of $f$ as
\begin{align} \label{eq: eucl. hess.}
f\left(Z+\rmd Z, Z^{*}+\rmd Z^{*}\right)&=f\left(Z, Z^{*}\right) 
+\left(\frac{\partial}{\partial \operatorname{vec}(Z)} f\left(Z, Z^{*}\right)\right) \rmd \operatorname{vec}(Z)+\left(\frac{\partial}{\partial \operatorname{vec}(Z^*)} f\left(Z, Z^{*}\right)\right) \rmd \operatorname{vec}\left(Z^{*}\right) \\
&\quad+\frac{1}{2}\left[\rmd \operatorname{vec}^{T}\left(Z^{*}\right) \rmd \operatorname{vec}^{T}(Z)\right]\left[\begin{array}{cc}
f_{ZZ^*}& f_{Z^*Z^*} \\
f_{ZZ}& f_{Z^*Z}
\end{array}\right]\left[\begin{array}{c}
\rmd \operatorname{vec}(Z) \\
\rmd \operatorname{vec}\left(Z^{*}\right)
\end{array}\right]+r\left(\rmd Z, \rmd Z^{*}\right) \, ,
\end{align}
where the higher order contribution $r\left(\rmd Z, \rmd Z^{*}\right)$ satisfies
\begin{align}
\lim _{\left(\rmd Z, \rmd Z^* \right) \rightarrow 0} \frac{r\left(\rmd Z, \rmd Z^*\right)}{\left\|\left(\rmd Z, \rmd Z^*\right)\right\|_{F}^{2}}=0 \, .
\end{align}

The second order derivatives are defined via
\begin{align*} 
f_{ZZ} = &= \frac{\partial}{\partial \, \mathrm{vec}(Z)^T} \frac{\partial}{\partial \, \mathrm{vec}(Z)} f(Z,Z^*,\dots; y) \, ,
\end{align*}
and similarly $f_{Z^*Z^*}$, $f_{Z^*Z}$ and $f_{ZZ^*}$. The vectorization is to be understood as joining together of indices in a fixed order. For instance $\operatorname{vec}:\mathbb{C}^{\ngates \times d^2 \times d \times d}\rightarrow \mathbb{C}^{\ngates d^4}$ vectorizes $\K$, where the individual $d$-dimensional legs are the matrix indices of the Kraus operators, and the $d^2$ index numbers the different Kraus operators. Note that $\operatorname{vec}\left(\frac{\partial}{\partial Z} f(Z,Z^*, y)\right) = \frac{\partial}{\partial \operatorname{vec}(Z)} f(Z,Z^*, y)$.

For the optimizations over 
$A$, $\K$, and $B$,
we need the first and second derivatives of $\mathcal{L}$ by the respective variables and their conjugates. 
Let 
$Z\in \{A,A^*,\K,\K^*,B, B^*\}$ and $Y \in \{Z,Z^*\}$. 
Then 
\begin{align*} 
\frac{\partial}{\partial Z} \mathcal{L}(Z,\dots; y) 
          &= \frac{2}{m} \sum_{\vec i} \left(p_{\vec i}(Z,\dots)-y_{\vec i}\right) \frac{\partial p_{\vec i}}{\partial Z} \, ,\\
\frac{\partial}{\partial Y} \frac{\partial}{\partial Z} \mathcal{L}(Z,Y,\dots; y)
          &= \frac{\partial}{\partial Y} \frac{2}{m} \sum_{\vec i} \left(p_{\vec i}(Z,Y,\dots)-y_{\vec i}\right) \frac{\partial p_{\vec i}}{\partial Z} \\
          &= \frac{2}{m} \sum_{\vec i} \frac{\partial p_{\vec i}(Z,Y,\dots)}{\partial Y} \frac{\partial p_{\vec i}(Z,Y,\dots)}{\partial Z} 
          + \frac{2}{m} \sum_{\vec i} \left(p_{\vec i}(Z,Y,\dots)-y_{\vec i}\right) \frac{\partial^2 p_{\vec i}(Z,Y,\dots)}{\partial Y \partial Z} \,  , 
\end{align*}
meaning that derivatives of the objective function reduce to the derivatives of the tensor $p$. 
Taking the derivative of a tensor network w.r.t.\ one of its constituent tensors can be easily done in the pictorial representation by removing the respective tensor. 
For instance, $\frac{\partial p_{\vec{i}}}{\partial \K^*}$ can be calculated as shown in Figure~\ref{fig:derivative}, using the product rule. Care has to be taken for the order of open indices when removing a tensor. In practice, we do not calculate the full tensor $\frac{\partial p}{\partial \K^*}$ of size $\ngates^{\seqlength}$ and only compute $\frac{\partial p_{\vec i}}{\partial \K^*}$ for $\vec{i} \in I$, since usually $|I| \ll n^\seqlength$. 

\subsection{Mean variation error dependence on the choice of objective function} \label{app:MLE}

\begin{figure}[htb]{}
    \includegraphics{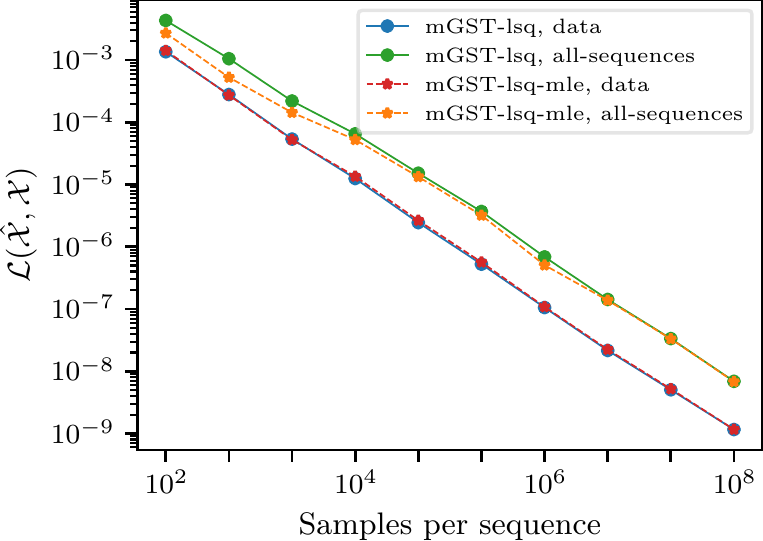}
    \hspace{40pt}
    \includegraphics{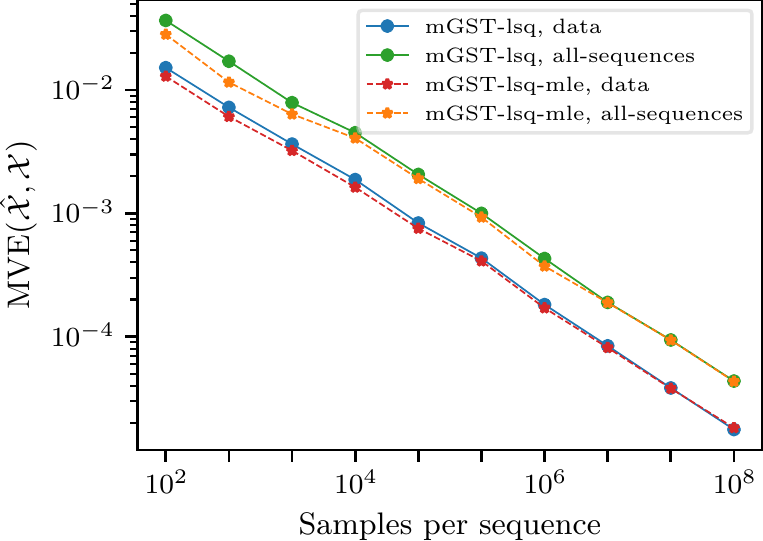}
  \caption{Effect of optimizing the log-likelihood function after the least squares objective function. The results of only the least squares optimization are denoted by \alg-lsq and those of additional log-likelihood optimization by \alg-lsq-mle. 
  The \textbf{left} plot shows the least squares objective function $\mc L(\hat{\mc X},\mc X)$ on the data sequences (used for estimation), as well as on all sequences of a given length $l$ (here $l = 7$). 
  The \textbf{right} plot uses the \ac{MVE}, again on data sequences or all sequences. \newline
  The underlying gate set is given by the XYI model with depolarizing noise of strength $p = 0.01$ on each gate and $p = 0.01$ on the initial state, as well as random unitary rotations $\e^{\i\gamma H}$ with $H \sim \mathrm{GUE}$ and $\gamma = 0.01$ on each gate. To ensure \alg-lsq is fully converged, we set the desired relative precision to $\epsilon = 10^{-5}$ in the convergence criterion, cp.\ Eq.~\eqref{eq: relative precision epsilon}.
  }
  \label{fig:MLE improvement}
\end{figure} 
A well motivated alternative to the least squares objective function defined in Eq.~\eqref{eq:objective function} is the likelihood function 
\begin{align}
L_I(A,\mathcal{K},B|y) \coloneqq \prod_{\vec i \in I} \prod_{j \in [n_E]} p_{j|\vec i}(A,\mathcal{K},B)^{k_{j|\vec i}} \, ,
\end{align}
where $y_{j| \vec i} = k_{j| \vec i}/m$ denotes again the relative number out of $m$ times outcome $j$ was measured for sequence $\vec i$, see also Eq.~\eqref{eq:y_and_k}. 
The likelihood function at a given model parametrization $(A,\mathcal{K},B)$ and for measurement results $y$ is precisely the probability of observing $y$, given the model probabilities $p(A,\mathcal{K},B)$. To simplify the optimization, often the logarithm of the likelihood function is chosen, since it shares the same maxima.
This log-likelihood function, which is also used in the final optimization procedure of \pyGSTi\ \cite{Nielsen2020GateSetTomography} is then given by
\begin{align} \label{eq:mle objective function}
\mathrm{log} L_{I}(A, \mathcal{K}, B| y)\coloneqq m \sum_{\vec i \in I} \sum_j y_{j|\vec i} \, \mathrm{log}\left[p_{j|\vec i}(A,\K,B)\right] \, .
\end{align}
In Figure~\ref{fig:MLE improvement} we show the effects of augmenting \alg\ (which is by default run on the least squares objective function for numerical reasons) with the log-likelihood function after a least squares estimate was found. 
To do this, we use Algorithm~\ref{Alg:main} with the negative log-likelihood function~\eqref{eq:mle objective function} as objective function. 

We observe that for the XYI-gate set, which we use to compare \alg\ and \pyGSTi\ in the main text, optimizing the log likelihood function decreases the mean variation error on the data sequences as well as on all sequences of the same length. The improvement is stronger for fewer samples and becomes negligible at around $10^6$ samples. 
Interestingly, log-likelihood optimization also improves the least squares error on all sequences, at the cost of slightly increasing it on the data sequences. 
This indicates that in the sample count range of $10^2 - 10^6$, optimizing the log-likelihood function leads to less overfitting.

\subsection{Noise-mitigation of shadow estimation with GST characterization} \label{app:Shadows}
In Section~\ref{sec: shadows} we numerically demonstrated how the results of a low-rank \ac{GST} experiment can be used to correct estimation protocols based on inverting an informationally complete POVM. Such protocols are recently referred to as shadow estimation \cite{Huang2020Predicting}. We now give a short mathematical description of the method and explain how low rank \ac{GST} estimates can be included in a scalable way. 

In the following, we use
bra-ket notation also for the space of linear operators $\mathrm{L}(\mathcal{H})$
and its dual space as defined by the canonical isomorphism induced by the Hilbert Schmidt inner product $\braketr{O}{\rho} = \Tr(O^{\dagger}\rho)$.
The quantum channel of a unitary $U$ is written by the corresponding calligraphic letter, e.g.\ $\mathcal{U} \ketr{\rho} \equiv \ketr{U\rho U^{\dagger}}$.

We consider the task of estimating the expectation value of multiple observables in an unknown quantum state that we can repeatedly prepare on a quantum device.
Being able to measure an informationally complete POVM $\{\Pi_x \}$ on the state, one can construct an estimator for an observable $O$.
\emph{Informationlly completeness} is equivalent to, in mathematical terms, the POVM constituting a \emph{frame} for $\mathrm L(\mathcal H)$, and the associated frame operator $\mathcal M = \sum_x \oketbra{\Pi_x}{\Pi_x}$ being invertible, see e.g.~Ref.~\cite{waldron2018introduction}.
We can calculate the canonical dual frame to the POVM 
as $|\tilde{\Pi}_x) = \mathcal M^{-1} \ketr{\Pi_x}$.
By construction we have the frame duality relation
\begin{equation}
    \sum_x \oketbra{\tilde \Pi_x}{\Pi_x} = \Id_{\mathrm L(\mathcal H)}.
\end{equation}
Thus, for any state $\rho$,
\begin{equation}\label{eq:insert_id}
    \obraket{O}{\rho} = \sum_x \obraket{O}{\tilde \Pi_x} \obraket{\Pi_x}{\rho}\,.
\end{equation}
By Born's rule, repeated measurements of the POVM yield i.i.d.\ samples $\Omega = (x_1, \ldots, x_m)$ from the distribution with density $p_\rho(x) = \obraket{\Pi_x}{\rho}$.
Given $\Omega$ we can calculate the empirical mean estimator 
\begin{equation}\label{eq:lin-estimator}
    \hat o = \frac1{
    |\Omega|} \sum_{x \in \Omega} \obraket{O}{\tilde \Pi_x}\,,
\end{equation}
and by \eqref{eq:insert_id}, $\EE[\hat o] = \obraket{O}{\rho}$.

The sequence of dual frame elements $(\tilde \Pi_{x_1}, \ldots, \tilde \Pi_{x_m})$ given by the measured samples $\Omega$ has been called the \emph{classical shadow} of $\rho$ in Ref.~\cite{Huang2020Predicting}.

A practical implementation of an informationally complete POVM on a digital quantum computer can be realized with measurements in randomly selected bases from a sufficiently large group.
To be explicit, we will consider the simplest and perhaps most
well-known example: the measurement in a randomly chosen multi-qubit Pauli-basis.
The POVM can be implemented by applying a random (different) local Clifford rotation on every qubit and measuring in the computational basis.
For informational-completeness it is sufficient to choose the rotations uniformly from the set $C = \{\Id, H, HS\}$, where $H$ is the Hadamard gate and $S$ the phase gate. 
In our notation, we consider POVM effects $\Pi_x = \Pi_{\mathbf g, \mathbf b} = \otimes_l \Pi_{g_l, b_l}$ 
indexed by $C^n \times \{0,1\}^n$ that are the tensor products of the local POVM effects $\Pi_{g_l, b_l} = \frac 13 g_l^\dagger \ketbra {b_l} {b_l} {g_l}$ with $g_l \in C$. 
Let $\{\hat\sigma_k \mid k \in \{0,1,2,3\}\}$ denote the Pauli matrices normalized in Frobenius norm.
The frame operator is given by
\begin{align}
  3^n \mathcal{M} &= \frac{1}{3^n}\left(\sum_{{U} \in \{\mathds{1}, H, HS\}}\mathcal{U}^{\dagger} (\ketr{\hat\sigma_0}\brar{\hat\sigma_0} + \ketr{\hat\sigma_3}\brar{\hat\sigma_3}) \mathcal{U}\right)^{\otimes n} \\
  &= \frac{1}{3^n}\left(3\ketr{\hat\sigma_0}\brar{\hat\sigma_0}  + \sum_i\ketr{\hat\sigma_i}\brar{\hat\sigma_i}\right)^{\otimes n} \\
  &= \begin{pmatrix} 1 & 0 & 0 & 0 \\ 0 & \frac{1}{3} & 0 & 0 \\ 0 & 0 & \frac{1}{3} & 0 \\ 0 & 0 & 0 & \frac{1}{3} \\\end{pmatrix}^{\otimes n},
\end{align}
where the matrix in the last line is represented in the Pauli-basis. 
Since $\mathcal M ^{-1}$ acts on qubit $l$ as $\mathcal M_l^{-1} (X) = 3X - \Tr(X) \mathds{1}$ for any $X$, we find $\tilde\Pi_{g_l, b_l} = \bigotimes_{i=1}^n \left(3 g_l^\dagger \ketbra {b_l} {b_l} {g_l} - \mathds{1}\right)$ \cite{Huang2020Predicting}.

Ref.~\cite{Huang2020Predicting} showed that when using random Pauli basis measurements, the variance of the mean estimator for estimating local observables does not scale with the system size.
Using a median-of-means estimator to boost the confidence, Ref.~\cite{Huang2020Predicting} further establishes that the expectation value of $M$ different $k$-local observables can be estimated
to $\epsilon$-additive precision from $\mathcal{O}(\log(M)4^k/\epsilon^2)$ state copies.

Experimental implementations of the POVM are prone to errors, effectively implementing a noisy POVM with effects $\Pi^\natural_x$.
In the envisioned implementation here, 
noise sources effect 
the implementation of the gates $C$ and
the noise induces a bias in the estimator for the observables.
However, if the noise is characterized to some extent we can correct the estimators for this bias.
To this end, let $\mathcal M^\natural$ be the (half-sided) noisy frame operator $\mathcal M^\natural = \sum_x \oketbra {\Pi_x} {\Pi^\natural_x}$.
If we know $\mathcal M^\natural$ in the classical post-processing, we can calculate a dual frame to $(\Pi^\natural_x|$ by $|\tilde\Pi^{\natural}_x) = \mathcal {M^\natural}^{-1} \ketr{\Pi_x}$.
Note that using the `half-sided noisy' frame operator instead of the frame operator of the noisy POVM yields an expression of a dual frame in terms of the ideal POVM and not the noisy POVM.
Using $\{\tilde\Pi^{\natural}_x\}$ instead of the ideal dual-frame in \eqref{eq:lin-estimator} yields unbiased estimators of observables even in the presence of noise, thus, effectively mitigating the noise.

This motivates our approach to noise mitigated shadow estimation. 
Having extracted a noise model via gate set tomography, we can numerically estimate $\mathcal M^\natural$ and, thus, construct (approximately) unbiased estimators. Our method is summarized in Protocol \ref{prot:Shadows} below. 

\begin{algorithm}
\SetAlgorithmName{Protocol}{List of protocols}
\TitleOfAlgo{}
\SetAlgoRefName{1}
\caption{GST-mitigated shadow estimation} \label{prot:Shadows} 
\SetKwInOut{Input}{input}{}{}
\Input{Target observable $O$, native local gate set $\mathcal{G}$ with $C \subseteq \mathcal{G}$}
Perform 2-qubit-\alg\ on implementation of $\mathcal{G}$ for qubit pairs $(1,2), \dots (n-1, n)$ \tikzmark{right} \tikzmark{top} 
\\ 
Gauge optimize \alg\ estimators to unitary target gates \tikzmark{bottom}\\
\For{$i \in [N]$\tikzmark{top2} }{
Select setting $\vec g \in C^{\otimes n}$ uniformly at random\\
Measure $\mathcal{U}_{\vec g} \ketr{\rho}$ in the standard basis\\
Save setting $\vec g$ and outcome $\vec{b}$ \tikzmark{bottom2}
}
Construct $\mathcal {M^\natural}$ from \alg\ gate estimates  \tikzmark{top3}\\
Compute single shot estimators $\left\{\hat{o}_i = \brar{O}\mathcal {M^\natural}^{-1} \ketr{\Pi_{\vec g, \vec b}}\right\}_{i = 1}^N$ \\
\Return{\tikzmark{bottom3} $\hat{O} = \operatorname{mean\ \textit{or}\ median-of-means}(\{o_i\})$} \tikzmark{bottom3}
\AddNote{top}{bottom}{right}{GST}
\AddNote{top2}{bottom2}{right}{Cl. Shadows data acquisition}
\AddNote{top3}{bottom3}{right}{Combined post processing}
\end{algorithm}

The results in Section \ref{sec: shadows} demonstrate our scheme numerically in simple but already practically relevant settings that we describe in the following.
When the multi-qubit unitaries and computational basis measurement implementing the POVM factorize into local tensor products, so does the ideal frame operator $\mathcal M$ and the dual frame (shadow).
But due to correlated gate-dependent noise, $\mathcal M^\natural$ might not exhibit this computationally tractable structure.
Furthermore, characterizing the implementation of exponentially many multi-qubit unitaries and the basis measurements without additionally assumption is infeasible.
In practice, however, noise-induced correlations and crosstalk might still predominantly affect a limited number of qubits simultaneously.
For example,
when noise predominantly affects neigboring qubits, we can use the implementation $\mathcal X$ 
of a gate set including $C\times C$ on neighboring qubits extracted via mGST to calculate $\mathcal M^\natural$.
To this end, let $\mathcal G^{(i,i+1)}_{g_1,g_2}$ denote the two-qubit process implementing the gate $g_1 \times g_2$ on qubit $i$ and $i+1$. For simplicity we ignore errors in the computational basis measurement.
We set $\Pi^{\natural, (i,i+1)}_{(g_1, g_2), (b_1, b_2)} = \frac19 (\mathcal G^{(i,i+1)}_{g_1, g_2})^\dagger \ketbra{b_1,b_2}{b_1,b_2} (\mathcal G^{(i,i+1)}_{g_1, g_2})^\dagger$ 
and numerically calculate
\begin{equation}
    \mathcal M^{\natural}_{i,i+1} = \sum_{g_1, g_2 \in C, b_1, b_2 \in \{0,1\}} \oket{\Pi_{g_1, b_1}} \oketbra{\Pi_{g_2, b_2}}{\Pi^{\natural, (i,i+1)}_{(g_1, g_2), (b_1, b_2)}}\,.
\end{equation}
This amounts to calculating a $16 \times 16$ matrix in the Pauli-basis that can be easily inverted.
The noise-mitigated single-shot estimators, thus, read
\begin{equation}
    \obraket{O}{\tilde \Pi^\natural_{\mathbf g, \mathbf b}} = \obra{O} {\bigotimes_{i=1}^{n/2}  (M^{\natural}_{2i,2i+1})^{-1}}\oket{\Pi_{g_{2i}, b_{2i}}}\oket{\Pi_{g_{2i+1}, b_{2i+1}}}\,.
\end{equation}
With this expression at hand, the rest of the protocol consists of computing the mean or median-of-means from a collection of single shot estimators, following the standard method of shadow estimation \cite{Huang2020Predicting}.

Ref.~\cite{2021PRXQ....2c0348C} proposes a complimentary approach for robust shadow-estimation, also inferring an approximation of the noisy frame operator from a separate calibration experiment.
Under the assumption of gate-independent noise, the authors derive a $2^n$ parameter expression for $\mathcal {M^\natural}$ as a Pauli-noise channel and devise (SPAM-robust) RB-style experiments to learn arbitrarily many of its parameters, where each parameter corresponds to an irreducible representation of the local Clifford group.
We find, in the gate-dependent noise model used here, that the frame operators significantly deviate from being a Pauli-noise channel. 
For this reason, this particular setting is more amenable to GST-mitigated shadows than to the protocol of Ref.~\cite{2021PRXQ....2c0348C}.
The plots on the left in Figure~\ref{fig:Shadows} show that already a typical $n=2$ frame operator with gate-dependent noise does not adhere to being diagonal in Pauli basis.

Ultimately, we envision that different robust and self-consistent noise and error characterization protocols, such as mGST for local coherent errors and RB for incoherent noise strength in different irreducible representations, can be combined to arrive at accurate and scalable estimates of the effective frame operator in the presence of noise.

\vfill 
\break 
\twocolumngrid

\let\oldaddcontentsline\addcontentsline 
\renewcommand{\addcontentsline}[3]{}

\section*{Acronyms} 
\hypertarget{Acronyms}{}
\bookmark[level=section,dest=Acronyms]{Acronyms}
\vspace{1pt}

\begin{acronym}[POVM]\itemsep.5\baselineskip
\acro{AGF}{average gate fidelity}

\acro{BOG}{binned outcome generation}

\acro{RB}{randomized benchmarking}

\acro{CP}{completely positive}
\acro{CPT}{completely positive and trace preserving}

\acro{DFE}{direct fidelity estimation} 

\acro{GST}{gate set tomography}

\acro{GUE}{Gaussian unitary ensemble}

\acro{HOG}{heavy outcome generation}

\acro{mGST}{manifold \acs{GST}}
\acro{MLE}{maximum likelyhood estimation}
\acro{MPS}{matrix product state}
\acro{MSVE}{mean squared variation error}
\acro{MUBs}{mutually unbiased bases} 
\acro{MVE}{mean variation error}
\acro{MSE}{mean squared error}
\acro{MW}{micro wave}

\acro{NISQ}{noisy and intermediate scale quantum}

\acro{POVM}{positive operator valued measure}
\acro{PVM}{projector-valued measure}

\acro{QAOA}{quantum approximate optimization algorithm}

\acro{SFE}{shadow fidelity estimation}
\acro{SFN}{saddle free Newton}
\acro{SIC}{symmetric, informationally complete}
\acro{SPAM}{state preparation and measurement}

\acro{QPT}{quantum process tomography}

\acro{rf}{radio frequency}

\acro{TT}{tensor train}
\acro{TV}{total variation}

\acro{VQE}{variational quantum eigensolver}

\acro{XEB}{cross-entropy benchmarking}


\end{acronym}

\end{document}